%% file: main.tex
\documentclass[12pt]{article}
\usepackage{bqf-paper}
\usepackage{algorithm}
\usepackage{algpseudocode}
\usepackage{amsmath}
\usepackage[utf8]{inputenc} 
\usepackage[T1]{fontenc}    
\usepackage{hyperref}       
\usepackage{url}            
\usepackage{booktabs}       
\usepackage{amsfonts}       
\usepackage{nicefrac}       
\usepackage{microtype}      
\usepackage{xcolor}         
\usepackage{tikz}
\usetikzlibrary{arrows.meta,positioning,calc,decorations.pathreplacing,shapes.geometric}
\usepackage[table]{xcolor}
\usepackage{listings}

\newcommand{\ignore}[1]{{}}

\title{\paperTitleSmallCaps{
New Lower and Upper Bounds for the Grothendieck Constant
}}

\author{
  \normalsize
  \setlength{\tabcolsep}{15pt}
  \begin{tabular}{ccc}
    Rahul Saha\footnotemark[1] & Alan Li\thanks{Equal contribution.}  & Anton Xue \\
    \small\texttt{rahul.saha@utexas.edu} & \small\texttt{alanli@cs.utexas.edu} & \small\texttt{anton.xue@utexas.edu} \\[3ex]
    Swarat Chaudhuri & Adam Klivans & Pravesh K. Kothari \\
    \small \texttt{swarat@cs.utexas.edu} & \small\texttt{klivans@cs.utexas.edu} & \small\texttt{kothari@cs.princeton.edu} \\[3ex]
     & Raghu Meka &  \\
     & \small\texttt{raghum@cs.ucla.edu} & \\[3ex]
  \end{tabular}
}

\date{}
\begin{document}

\maketitle

\input{paper/abstract}

\input{paper/introduction}

\input{paper/technical_overview}

\newpage

\part{The Upper Bound}

\input{paper/krivine}

\input{paper/limiting_schemes}

\input{paper/asymptotic_scheme}

\input{paper/coefficient_formula}

\input{paper/certification}

\newpage
\part{The Lower Bound}
\label{part:LowerBound}

\input{paper/lower_bound_general}

\input{paper/statement_ai_use}

\bibliographystyle{amsalpha}
\bibliography{references}

\input{paper/appendix}

\end{document}

%% file: paper/abstract.tex
\abstract{

We establish new bounds on the Grothendieck constant \(\KG\):
\[
    \frac{6\pi}{11}
    \le
    \KG
    \le
    \frac{\pi}{2\log(1+\sqrt2)} - 10^{-4}.
\]
Methodologically, our lower bound approach differs from previous works by establishing limitations on the asymptotically optimal Krivine schemes, rather than giving explicit constructions of gap instances.
Our upper bound is obtained by proposing and analyzing the first asymptotic construction of rounding schemes, whereas previous works only consider low-dimensional schemes.
Together, these bounds determine the previously unknown tenths digit of \(\KG\) to be $7$. 
The bounds were discovered by a long-running collaborative effort of humans and a long-horizon AI research system that we engineered. A detailed account of the AI-assisted research process appears in a companion paper~\cite{lisaha2026longhorizon}.

}
\newpage 

\tableofcontents

%% file: paper/introduction.tex
\section{Introduction}

The Grothendieck inequality~\cite{Gro53,LP68} states that there exists a universal constant $K<\infty$ such that for every $m,n\in\N$ and every matrix $A=(a_{ij})\in\R^{m\times n}$,
\begin{equation}\label{eq:grothendieck-ineq}
  \max_{\eps_i,\delta_j\in\{\pm1\}}
  \left|\sum_{i=1}^m\sum_{j=1}^n a_{ij}\eps_i\delta_j\right|
  \le
  \sup_{u_i,v_j\in S^{m+n-1}}
  \left|\sum_{i=1}^m\sum_{j=1}^n a_{ij}\ip{u_i}{v_j}\right|
  \le
  K\max_{\eps_i,\delta_j\in\{\pm1\}}
  \left|\sum_{i=1}^m\sum_{j=1}^n a_{ij}\eps_i\delta_j\right|.
\end{equation}
Here $S^{m+n-1}$ denotes the unit sphere in $\R^{m+n}$.  The Grothendieck constant $\KG$ is the infimum of $K$ for which this inequality holds.

The inequality is equivalent to the statement that the integrality gap of the canonical semidefinite relaxation for maximizing the bilinear form $x^\top Ay$ over $x\in\{-1,1\}^m$ and $y\in\{-1,1\}^n$ is bounded by an absolute constant.
This connection leads to efficient algorithms for approximating combinatorial problems, such as matrix cut norms~\cite{AN06}.  Beyond this algorithmic role, Grothendieck's inequality is central in the geometry of Banach spaces, operator spaces, harmonic analysis, and quantum information; see, for example, the surveys and monographs~\cite{KN11,pisier2011grothendieckstheorempastpresent,DFS08}.

However, the exact value of $\KG$ is still unknown.
On the lower bound, existing strategies aim to construct matrices for which the vector optimum is large relative to the sign optimum.
While Grothendieck's original work implies that $\KG \geq \pi/2$~\cite{Gro53}, this was substantially improved by Davie~\cite{Davie1984LowerBoundKG} and Reeds~\cite{Reeds1991LowerBoundKG}, who independently used a high-dimensional Gaussian-operator construction to prove
\[
\KG
\geq K_{\mathrm{DR}} = 1.6769\ldots
\]
The Davie--Reeds inequality was the best-known result for more than four decades, until two recent works by
Heilman~\cite{heilman2026lowerboundgrothendiecksconstant} and Jones and Malavolta~\cite{jones2026grothendieckconstantstrictlylarger} independently showed that $\KG \geq K_{\mathrm{DR}}+10^{-26}$
and $\KG \geq K_{\mathrm{DR}}+10^{-12}$, respectively.





On the upper bound, Krivine's seminal result~\cite{Kri77} proved that
\begin{equation}\label{eq:krivine-bound-intro}
  \KG\le {\pi\over 2\log(1+\sqrt2)}=1.7822\ldots
\end{equation}
The proof is constructive: it gives a rounding scheme that maps input vectors $u_i,v_j$ to signs $\eps_i,\delta_j\in\{\pm1\}$ after an infinite-dimensional preprocessing step and a random hyperplane rounding step.
Krivine also conjectured that this bound was tight, and this stood for over three decades until it was disproven by Braverman, Makarychev, Makarychev, and Naor in 2011; specifically, they ``mixed'' the hyperplane rounding scheme with a non-hyperplane one to establish the strict inequality $\KG < \tfrac{\pi}{2\log(1+\sqrt2)}$~\cite{braverman2011grothendieckconstantstrictlysmaller}.
Interestingly, however, Naor and Regev subsequently showed that mixed Krivine schemes (\cref{sec:krivine-rounding-scheme}) are in fact optimal: in the limit of dimension, there exists such a scheme that yields arbitrarily good approximations of $K_G$~\cite{naor2014krivine}.
This thus opened a strategy for attacking the strict inequality's unspecified numerical gap, and was leveraged in recent works by Heilman~\cite{heilman2026upperboundgrothendiecksconstant} and Li et al.~\cite{LiSahaXueChaudhuriKlivansKothariMeka2026} to independently arrive at explicit numerical improvements on the order of $10^{-5}$.

\paragraph{This work.}
In this paper, we make the following contribution:
\begin{enumerate}
    \item \textbf{Upper bound.} We introduce \emph{limiting Krivine schemes}, an analytic framework for constructing and studying asymptotic families of increasingly high-dimensional rounding schemes. Using this framework, we construct a new ``cubic--quintic'' scheme that yields
    \[
        K_G
        \le {\pi\over 2\log(1+\sqrt2)} - 3.47 \times 10^{-4}
        = 1.7818 \ldots.
    \]
    This differs from previous works, which only considered one- or two-dimensional rounding schemes, and in particular positively answers a question posed by Braverman et al.~\cite{braverman2011grothendieckconstantstrictlysmaller} about whether increasing dimension can lead to improved bounds on $K_G$.

    \item \textbf{Lower bound.} We find a new approach to proving lower bounds on $K_G$ and show:
    \[
        K_G \ge \frac{6\pi}{11}=1.7135\ldots.
    \]
    Together with our upper bound, this determines the tenths digit of \(K_G\) to be 7.
    Notably, our approach is to show a fundamental limitation of Krivine schemes and then apply a result from Naor and Regev~\cite{naor2014krivine}.
    

    \item \textbf{AI methodology.}
    We further detail our use of AI in mathematics research in a companion paper~\cite{lisaha2026longhorizon}.
    Because mathematics research is a long-horizon task that comprises complex goals, persistent memory, and delayed rewards, it is often unclear how to effectively use AI.
    Our work thus provides an extensive case study on how to develop and interact with emerging AI systems.    
    Critically, this system is also highly collaborative with human interaction, which allows it to effectively explore promising proof strategies.

\end{enumerate}

%% file: paper/technical_overview.tex
\section{Technical Overview}

Here, we provide the background and exposition of the main ideas in our work.

\subsection{The upper bound}

\paragraph{Classical Krivine rounding schemes.}
A classical Krivine scheme consists of odd sign functions \(f,g:\mathbb R^k\to\{-1,1\}\). Its normalized correlation function is
\(
    H_{f,g}(t):=\frac{\pi}{2}\mathbb E[f(X)g(Y)],
\)
where \(X,Y\) are standard Gaussian vectors with coordinatewise correlation \(t\). Suppose \(H^{-1}(z)=\sum_{n\geq 1}a_nz^n\), and let its \emph{majorant} series be  \(M(s):=\sum_{n\geq 1}|a_n|s^n\). Krivine's tensor preprocessing
(\cref{app:krivine-preprocessing})
gives
\(
    M(\gamma)\leq 1
    \implies 
    K_G\leq\frac{\pi}{2\gamma}.
\)
The upper-bound problem is therefore to construct a correlation function whose inverse has a large admissible  $\gamma$ so that $M(\gamma) = 1$.

\paragraph{Limiting Krivine rounding schemes.}
In ordinary Krivine schemes, we need individual Gaussian coordinates $X_i, Y_i$ to have coordinate-wise correlation $t$. We enlarge the search space by allowing $X_i, Y_i$ to have coordinate-wise correlation
\[
    \rho_i(t)=\sum_{\substack{d\geq 1\\ d\ \mathrm{odd}}}c_{i,d}t^d,
    \qquad
    \sum_d|c_{i,d}|\leq 1.
\]
This gives us more degrees of freedom in designing \(H\) whose associated \(M\) has larger \(\gamma\) satisfying $M(\gamma) = 1$.





\paragraph{Approximation of limiting schemes using classical schemes.}
While we may design the correlations between variables to admit any \(\rho\) of the above structure, it is not immediately clear how the rounding algorithm works. 
Here, we show that any limiting scheme can in fact be approximated arbitrarily well by classical Krivine schemes, whose rounding algorithm is established. 
Therefore, it suffices to find a limiting scheme whose inverse majorant $M$ admits a high $\gamma$ with $M(\gamma) = 1$.

More precisely, we construct a sequence $\mathcal{S}_1, \mathcal{S}_2, \ldots$ of classical Krivine schemes of increasing dimension \(k\), whose correlation functions $H_{\mathcal{S}_1}, H_{\mathcal{S}_2}, \dots$ converge to the correlation function $H_{\mathcal{S}}$ of a limiting scheme $\mathcal{S}$.
The main idea is to use the central limit theorem to collapse the increasing number of Gaussian random variables in the classical schemes \(\mathcal{S}_k = (f_k, g_k)\), where \(f_k, g_k : \mathbb{R}^k \to \{\pm 1\}\).





%

\paragraph{The cubic--quintic scheme.}
A limiting Krivine scheme that attains the paper's stated upper bound on \(K_G\) is defined as follows.
For suitable parameters \(s_3,s_5,\vartheta>0\), set
\[
    \rho(t):=
    \frac{t-s_3^2t^3+s_5^2t^5}{1+s_3^2+s_5^2},
\]
and use the partitions
\[
    f(w,x)=\operatorname{sgn}\!\left(w+\vartheta\operatorname{He}_3(x)\right),
    \qquad
    g(w,x)=\operatorname{sgn}\!\left(w-\vartheta\operatorname{He}_3(x)\right).
\]
We set the first coordinate \(w\) to have correlation \(\rho(t)\), and set the second coordinate \(x\) to have correlation \(t\).



\paragraph{Deriving the upper bound.}
In order to prove the upper bound (which we denote as $K_G^{UB}$), one needs to show that the inverse majorant $M(\gamma) \leq 1$ holds, for $\gamma=\frac{\pi}{2K_G^{UB}}$.
Doing this is non-trivial, and requires both an analytic step and a computational step. 

For the analytic step, write
\(H(t)=b_1t+\sum_{m\geq 3}b_mt^m\), and let
\(\Delta_H:=\sum_{m\geq 3}|b_m|\) denote its total nonlinear coefficient mass. We prove that it is sufficient to show 
\[
    \gamma+\Delta_H<b_1.
\]
Thus, rather than computing the inverse series \(H^{-1}\), it is enough to certify a lower bound on the linear coefficient \(b_1\) and an upper bound on \(\Delta_H\).
We then compute the first terms of $\{b_n\}$ using a formula involving the Hermite coefficients of $f$ and $g$, which are one-dimensional Gaussian integrals. Finally, we bound the tail of $\Delta_H$ through a single boundary-norm estimate for \(D^3H\), where \(D=t\frac{d}{dt}\). Rigorous interval arithmetic verifies the resulting finite collection of inequalities, yielding the claimed upper bound on \(K_G\).


\subsection{The lower bound}

\paragraph{Mixed Krivine schemes are optimal.}
Naor and Regev~\cite{naor2014krivine} proved that mixed Krivine schemes are asymptotically optimal. We use a quantitative consequence of their construction: there exist limiting mixed correlation functions \(H_k\) with admissible inverse-majorant parameters \(\gamma_k\) satisfying
\[
    \gamma_k\longrightarrow\frac{\pi}{2K_G}.
\]
Consequently, any universal upper bound on the admissible parameter of a Krivine correlation function gives a lower bound on \(K_G\).

\paragraph{The affine coefficient inequality.}
To obtain this upper bound, we ask which Taylor coefficients can actually arise from a Krivine correlation function. For odd sign functions \(f,g:\mathbb R^d\to\{-1,1\}\), write
\[
    H_{f,g}(t)=b_1t+b_3t^3+\cdots.
\]

Surprisingly, we show that there is the clean constraint on the first two coefficients:
\begin{equation}
    b_3\geq 2b_1-\frac{11}{6}.
    \label{eq:affine-coefficient-inequality}
\end{equation}
Its affine form is important, since it is preserved under mixtures and coefficientwise limits. 

Pushing the affine inequality to the inverse series $H^{-1}$ yields bounds on the inverse coefficients $a_1$ and $a_3$, which in turn gives $\gamma \leq \frac{11}{12}$. This results in the desired $K_G \geq \frac{6\pi}{11}$ .

\paragraph{Hermite expansion.} The approach to proving this is by looking at Hermite expansions of $f, g$. For a square-integrable function \(F\) on Gaussian space, let \(P_iF\) denote its orthogonal projection onto the Hermite polynomials of total degree \(i\).  If
\(H_{f,g}(t)=b_1t+b_3t^3+\cdots\) is the normalized correlation function of \(f\) and \(g\), then the basic orthogonality properties of Hermite polynomials give
\(b_1=\frac{\pi}{2}\langle P_1f,P_1g\rangle\) and
\(b_3=\frac{\pi}{2}\langle P_3f,P_3g\rangle\).
Our affine coefficient inequality is
\begin{equation*}
    b_3\geq 2b_1-\frac{11}{6}.
\end{equation*}
Since \(\nu:=\mathbb E|Z|=\sqrt{2/\pi}\), this is equivalent to
\[
    2\langle P_1f,P_1g\rangle-\langle P_3f,P_3g\rangle
    \leq \frac{11}{6}\nu^2.
\]

\paragraph{Agreement and disagreement.}
In order to prove the inequality, we make a key substitution. The substitution is to write
\(h:=(f+g)/2\) and \(k:=(f-g)/2\). The function \(h\) can be thought of as the \textit{agreement part}, where \(f\) and \(g\) agree, while \(k\) is the corresponding \textit{disagreement part}. Moreover,
\(h,k\in\{-1,0,1\}\) and \(h^2+k^2=1\).

Substituting \(f=h+k\) and \(g=h-k\) into the preceding inequality, and discarding a nonpositive term, it suffices to prove
\begin{equation}
    2\|P_1h\|^2-\|P_3h\|^2+\|P_3k\|^2
    \leq \frac{11}{6}\nu^2.
    \label{eq:agreement-disagreement-reduction}
\end{equation}
The first two terms depend only on the agreement function \(h\), while the final term measures the degree-three contribution of the disagreement function \(k\). We bound these two parts separately.

\paragraph{Agreement terms.}
Rotate coordinates so that \(P_1h\) lies in a single Gaussian direction \(S\), and write
\(x:=\|P_1h\|/\nu\in[0,1]\). The positive agreement term is then
\(2\|P_1h\|^2=2\nu^2x^2\).

To control the helpful negative term \(-\|P_3h\|^2\), average \(h\) over all coordinates transverse to \(S\). Namely, writing the Gaussian input as \((S,T)\), set
\(A(s):=\mathbb E_T[h(s,T)]\). Then \(|A|\leq1\) and
\(\mathbb E[A(S)S]=\nu x\). A one-dimensional rearrangement inequality shows that this prescribed linear moment forces a cubic Hermite moment:
\[
    \|P_3h\|^2
    \geq
    \frac{\nu^2}{6}\Delta_+(x)^2
\]
for an explicit function \(\Delta\). Thus, a large linear agreement term necessarily creates a compensating cubic agreement term.

\paragraph{Disagreement term.}
For each fixed transverse vector \(t\), consider the one-dimensional slice
\(u_t(s):=k(s,t)\). Since \(k\in\{-1,0,1\}\), each \(u_t\) is a ternary function on the Gaussian line. The degree-three Hermite decomposition implies
\[
    \|P_3k\|^2
    \leq
    \mathbb E_T V(u_T),
    \qquad
    V(u):=
    \sum_{j=0}^3
    \bigl(\mathbb E[u(Z)\psi_j(Z)]\bigr)^2.
\]
We control these first four Hermite moments through the weighted support
\(\beta(u):=\mathbb E[|Z|u(Z)^2]\). The key one-dimensional estimate is
\begin{equation}
    V(u)\leq 3\nu\beta(u)-\beta(u)^2.
    \label{eq:one-dimensional-support-inequality}
\end{equation}
Conceptually, this says that a ternary function cannot have both a support concentrated cheaply near the origin and a large projection onto all Hermite degrees at most three. The inequality admits an exact finite-dimensional dual reduction. After introducing a support price \(c\), the optimizing ternary function has the threshold form
\(
    u(z)=\operatorname{sgn}(p(z))
    \mathbf 1_{\{|p(z)|\geq c|z|\}},
\)
where \(p\) is a unit cubic Hermite polynomial. The resulting one-dimensional Gaussian integral inequalities are verified using outward-rounded interval arithmetic. Finally, \(h^2+k^2=1\) relates the average weighted support of the slices \(u_T\) to the agreement parameter \(x\). Combining this relation with
\eqref{eq:one-dimensional-support-inequality} gives an explicit upper bound on
\(\|P_3k\|^2\) in terms of \(x\). Together with the agreement estimate, this reduces
\eqref{eq:agreement-disagreement-reduction} to an elementary scalar inequality on
\(x\in[0,1]\), proving \eqref{eq:affine-coefficient-inequality}.

%% file: paper/krivine.tex
\section{Rounding Schemes}

In this section, we will define a new class of rounding schemes. But first, we recall classical Krivine schemes ~\cite{krivine1977constante}. 

\subsection{Classical Krivine rounding schemes}
\label{sec:krivine-rounding-scheme}

A Krivine rounding scheme is a pair of partitions of $\R^d$, $f,g: \R^d \rightarrow \{-1,1\}$, such that $f,g$ are odd functions. We will work with $f,g$ that satisfy some additional properties. To describe these, we first introduce some basic setup~\cite{braverman2011grothendieckconstantstrictlysmaller}. 


\begin{definition}[Gaussian correlation function]\label{def:gaussian-correlation}
Let $f,g:\mathbb{R}^k\to\{-1,1\}$ be odd measurable functions.
Let $X, Y$ be standard jointly Gaussian vectors in $\mathbb{R}^k$ satisfying
\[
\mathbb{E}[X_iX_j]
=
\mathbb{E}[Y_iY_j]
=
\delta_{ij},
\qquad
\mathbb{E}[X_iY_j]
=
\delta_{ij}t.
\]
The \emph{arcsine-normalized correlation function} associated to $(f,g)$ is
\[
H_{f,g}(t)
:=
\frac{\pi}{2}\,
\mathbb{E}\left[
f\left(X\right)
g\left(
Y
\right)
\right],
\qquad -1<t<1.
\]
\end{definition}

We will use the following basic properties of the Gaussian correlation function and its inverse $H_{f,g}^{-1}$ in our analysis (see \cite{braverman2011grothendieckconstantstrictlysmaller}).

\begin{lemma}\label[lemma]{lemma:correlation-profile-analytic}
Let $f,g:\mathbb{R}^k\to\{-1,1\}$ be odd measurable functions. Then $H_{f,g}$ is odd and
extends holomorphically to the strip
\[
S:=\{z\in\mathbb{C}: |\operatorname{Re}z|<1\}.
\]
In particular, if $H_{f,g}'(0)\neq 0$, then $H_{f,g}$ has a unique \textbf{local inverse branch} at
the origin, i.e., there are open neighborhoods
$U,V\subset\C$ of $0$ such that
\(
    H_{f,g}:U\to V
\)
is biholomorphic. We write
\(
    H_{f,g}^{-1}:V\to U
\)
for this local inverse, chosen so that $H_{f,g}^{-1}(0)=0$. This inverse is
unique when restricted to this branch, and is an odd function.
\end{lemma}


\begin{lemma}[Local Inverse Expansion]\label[lemma]{lemma:local-inverse-expansion}

Let $f,g:\mathbb{R}^k\to\{-1,1\}$ be odd and measurable functions with $H'_{f,g}(0) \neq 0$. Let $H^{-1}_{f,g}$ be the unique local inverse branch of $H_{f,g}$ at the origin. Then,
\begin{itemize}
\item $H^{-1}_{f,g}$ is odd. 
    \item $H^{-1}_{f,g}$ has a Taylor-series expansion at the origin: $H^{-1}_{f,g}(\zeta) = \sum_{n\geq 1} a_n \zeta^n$. 
\end{itemize}
Further, if $\gamma > 0$ is such that $\sum_{n \geq 1} |a_n| \gamma^n \leq 1$, then the Grothendieck constant 
$$ K_G \leq \frac{\pi}{2\gamma}.$$ 
\end{lemma}


With the above notation, we define the \textbf{inverse majorant series} for the rounding scheme $(f,g)$ as 
\[
M_{f,g}(s):=\sum_{n\ge 1}|a_n|s^n.
\]

We are now ready to define Krivine schemes. 

\begin{definition}[$\gamma$-admissible Krivine Rounding Scheme]
\label[definition]{def:gamma-admissible}
    For $f,g:\mathbb{R}^k\to\{-1,1\}$ odd and measurable functions with $H'_{f,g}(0) \neq 0$, and $\gamma > 0$, we say $(f,g)$ is a $\gamma$-admissible Krivine rounding scheme if the conditions of \cref{lemma:local-inverse-expansion} hold. 
\end{definition}
We note that Krivine~\cite{Kri77} used a $\gamma$-admissible Krivine rounding scheme to show that \(K_G \le \frac{\pi}{2\gamma}.\)

%% file: paper/limiting_schemes.tex
\subsection{Preliminaries}

We first state Vitali's theorem, which lets us convert pointwise convergence to uniform convergence. See~\cite{schiff1993normal} for a reference. 

\begin{fact}[Vitali's convergence theorem]
Let \(h_1,h_2,\ldots\) be holomorphic functions on a connected open set
\(\Omega\). Suppose that:
\begin{enumerate}
    \item the sequence \((h_N)_{N\geq 1}\) is locally bounded on \(\Omega\); and
    \item \(h_N(z)\) converges for every \(z\) in some subset of \(\Omega\)
    having an accumulation point in \(\Omega\).
\end{enumerate}
Then \((h_N)_{N\geq 1}\) converges locally uniformly on \(\Omega\) to a
holomorphic function.
\end{fact}

\subsection{Limiting Krivine rounding schemes}
\label{sec:LimitingKrivine}

In this section, we extend classical Krivine schemes by allowing them to have arbitrary correlations among variables (up to some criteria). We will show that we can state asymptotic statements about classical Krivine schemes using this extended class of schemes, that we call \textit{limiting Krivine schemes}. This gives us analytic machinery for searching and analyzing limits of high dimensional schemes that were previously inaccessible via standard techniques. It is first instructive for us to start with a motivating example. 

\paragraph{Motivating example.}
Consider the two-dimensional classical Krivine rounding scheme
\[
f(X_1,X_2)=\operatorname{sgn}(X_1+\lambda X_2),
\qquad
g(Y_1,Y_2)=\operatorname{sgn}(Y_1+\lambda Y_2),
\]
where \(X_1,X_2,Y_1,Y_2\) are standard Gaussians, \(X_1\) is independent of \(X_2\), \(Y_1\) is independent of \(Y_2\), and
\[
\operatorname{corr}(X_1,Y_1)
=
\operatorname{corr}(X_2,Y_2)
=
t.
\]
In a classical Krivine scheme, every pair of corresponding coordinates has correlation \(t\). Suppose instead that we would like the first pair to have correlation \(t\), but the second pair to have correlation \(t^3\). To obtain this from a classical high-dimensional scheme, define
\[
U_N=\frac{1}{\sqrt N}\sum_{j=1}^N \mathrm{He}_3(X_{2,j}),
\qquad
V_N=\frac{1}{\sqrt N}\sum_{j=1}^N \mathrm{He}_3(Y_{2,j}),
\]
where the pairs \((X_{2,j},Y_{2,j})\) are independent standard Gaussian pairs with correlation \(t\). The Hermite identity gives
\[
\mathbb{E}[U_NV_N]
=
\mathbb{E}[\mathrm{He}_3(X_{2,1})\mathrm{He}_3(Y_{2,1})]
=
t^3.
\]
Moreover, by the central limit theorem, \((U_N,V_N)\) converges in distribution to a standard Gaussian pair \((U,V)\) with correlation \(t^3\). Thus, the classical rounding schemes
\[
F_N=\operatorname{sgn}(X_1+\lambda U_N),
\qquad
G_N=\operatorname{sgn}(Y_1+\lambda V_N)
\]
converge to the fixed-dimensional Gaussian rounding scheme
\[
f(X_1,U)=\operatorname{sgn}(X_1+\lambda U),
\qquad
g(Y_1,V)=\operatorname{sgn}(Y_1+\lambda V),
\]
in which the two coordinate pairs have correlations \(t\) and \(t^3\), respectively.

To use this construction, we need to carefully justify the following analytic steps:
\begin{enumerate}
    \item Pointwise convergence of the correlation functions must be strengthened to locally uniform convergence;
    \item An inverse-majorant $\gamma$-admissibility proved for the limiting scheme must remain valid for all sufficiently large finite-dimensional approximations.
\end{enumerate}

We now generalize this construction and prove both properties. First, we define the notion of allowable correlation maps.  
\begin{definition}[Allowable correlation map]
\label{def:allowable-correlation-map}
A function \(\rho:(-1,1)\to\mathbb{R}\) is an \emph{allowable correlation map}
if it admits an expansion
\[
\rho(t)
=
\sum_{\substack{d\geq 1\\ d\ \mathrm{odd}}} c_d t^d,
\qquad
\sum_{\substack{d\geq 1\\ d\ \mathrm{odd}}}|c_d|
\leq 1.
\]
\end{definition}

We can now state the definition of limiting Krivine schemes. 

\begin{definition}[Limiting Krivine Schemes]
A $k$-dimensional \emph{limiting Krivine scheme} is a tuple $\mathcal{S} = (f, g, \rho)$, where  \(f,g:\mathbb{R}^k\to\{-1,1\}\) are odd measurable functions where the discontinuity sets have Gaussian measure zero, and \(\rho = (\rho_1,\ldots,\rho_k)\) is a tuple of allowable correlation maps.  

The arcsine-normalized correlation function of this scheme is given by 
\[
H_{f, g, \rho}(t)
=
\frac{\pi}{2}\mathbb{E}[f(X)g(Y)].
\]
where \(X,Y\in\mathbb{R}^k\) are jointly Gaussian such that
\[
\mathbb{E}[X_iX_j]
=
\mathbb{E}[Y_iY_j]
=
\delta_{ij},
\qquad
\mathbb{E}[X_iY_j]
=
\delta_{ij}\rho_i(t).
\]
\end{definition}

\begin{remark}
    This is a generalization of the usual high dimensional Krivine schemes with $\rho_i(t) = t$ for all $i$. 

    At this stage, \(H_{f, g, \rho}\) is defined only on \((-1,1)\); in particular, we do not assume that it extends holomorphically to the strip associated with classical Krivine schemes.
\end{remark}

The motivation for this definition is that every limiting Krivine scheme can be approximated arbitrarily well by ordinary finite-dimensional Krivine schemes, as we will show below. This allows us to use the usual projection scheme to round the vectors, and get the inverse-majorant machinery for free. We make this precise in the following theorem.  

\begin{theorem}[Finite-dimensional approximation of limiting schemes]
Let \(\mathcal S=(f,g,\rho_1,\ldots,\rho_k)\) be a limiting Krivine scheme. Then there exists a sequence of ordinary finite-dimensional Krivine schemes
\(
\mathcal S_1,\mathcal S_2,\ldots,
\)
of increasing dimension, whose correlation functions $H_{\mathcal S_1}, H_{\mathcal S_2}, \ldots$ converge to the correlation function $H_{\mathcal S}$ of \(\mathcal S\), uniformly on every compact subset of \((-1,1)\).
\end{theorem}

We dedicate the rest of the section to proving this theorem. First, we prove an equivalence condition on allowability of a function. 

\begin{lemma}[Allowable function criteria]
A map \(\rho\) is allowable if and only if there are odd functions
\(p,q: \R^3 \to \R\), such that
\(\|p\|_2=\|q\|_2=1\) and
\[
\rho(t)=\mathbb{E}[p(U)q(V)],
\]
where \(U,V\in\mathbb{R}^m\) are standard Gaussian vectors with
\(\mathbb{E}[U_iV_j]=t\delta_{ij}\).
\end{lemma}

\begin{proof}
Suppose first that \(p\) and \(q\) are odd and have \(L_2\)-norm one. Write
their Hermite expansions by degree as
\[
p=\sum_{\substack{d\geq 1\\ d\text{ odd}}}p_d,
\qquad
q=\sum_{\substack{d\geq 1\\ d\text{ odd}}}q_d.
\]
Here \(p_d\) and \(q_d\) are the degree-\(d\) Hermite components.

Different Hermite degrees are orthogonal, while degree \(d\) contributes a
factor \(t^d\). Hence
\[
\mathbb{E}[p(U)q(V)]
=
\sum_{\substack{d\geq 1\\ d\text{ odd}}}
\langle p_d,q_d\rangle t^d.
\]
Set \(c_d=\langle p_d,q_d\rangle\). Then
\[
\sum_d|c_d|
\leq
\sum_d\|p_d\|_2\|q_d\|_2
\leq
\left(\sum_d\|p_d\|_2^2\right)^{1/2}
\left(\sum_d\|q_d\|_2^2\right)^{1/2}
=
1.
\]
Thus the resulting correlation map is allowable.

Conversely, suppose that
\[
\rho(t)=\sum_{d\text{ odd}}c_dt^d,
\qquad
s:=\sum_{d\text{ odd}}|c_d|\leq 1.
\]
For \(x=(x_1,x_2,x_3)\in\mathbb{R}^3\), define
\[
p(x)
=
\sum_{d\text{ odd}}
\sqrt{|c_d|}\operatorname{He}_d(x_1)
+
\sqrt{1-s}\,x_2
\]
and
\[
q(x)
=
\sum_{d\text{ odd}}
\operatorname{sgn}(c_d)\sqrt{|c_d|}
\operatorname{He}_d(x_1)
+
\sqrt{1-s}\,x_3.
\]
Both functions are odd. Orthogonality gives \(\|p\|_2=\|q\|_2=1\).

Let \(U,V\in\mathbb{R}^3\) be standard Gaussian vectors with
\(\mathbb{E}[U_iV_j]=t\delta_{ij}\). The terms involving \(U_2\) and \(V_3\)
have zero cross-correlation. Therefore
\[
\mathbb{E}[p(U)q(V)]
=
\sum_{d\text{ odd}}c_dt^d
=
\rho(t).
\]
\end{proof}

We are now ready to prove our theorem, which follows as a direct corollary of the stronger statement we prove below. 

\begin{theorem}[Finite-dimensional approximation of limiting schemes]
Let \((f,g,\rho_1,\ldots,\rho_k)\) be a limiting Krivine scheme.
Then its correlation function extends holomorphically to the unit disk and is
the locally uniform limit there of the correlation functions of ordinary
finite-dimensional Krivine schemes.
\end{theorem}

\begin{proof}
By the above lemma, for each \(i\in\{1,\ldots,k\}\), there exist odd measurable functions 
\(
p_i,q_i:\mathbb{R}^3\to\mathbb{R}
\)
with $\|p\|_2 = \|q\|_2 = 1$ such that, whenever \(U,V\in\mathbb{R}^3\) are joint standard Gaussian vectors with
\(
\mathbb{E}[U_aV_b]=t\delta_{ab},
\)
we have
\[
\mathbb{E}[p_i(U)q_i(V)]
=
\rho_i(t).
\]
We combine these functions into vector-valued maps. Write
\(u\in\mathbb{R}^{3k}\) in \(k\) blocks as
\[
u=\bigl(u^{(1)},\ldots,u^{(k)}\bigr),
\qquad
u^{(i)}\in\mathbb{R}^3,
\]
and define
\[
P(u)
=
\bigl(
p_1(u^{(1)}),\ldots,p_k(u^{(k)})
\bigr),
\]
\[
Q(u)
=
\bigl(
q_1(u^{(1)}),\ldots,q_k(u^{(k)})
\bigr).
\]
For each positive integer \(N\), define functions on
\(\bigl(\mathbb{R}^{3k}\bigr)^N\) by
\[
F_N(u_1,\ldots,u_N)
=
f\left(
\frac{1}{\sqrt{N}}\sum_{r=1}^N P(u_r)
\right)
\]
and
\[
G_N(v_1,\ldots,v_N)
=
g\left(
\frac{1}{\sqrt{N}}\sum_{r=1}^N Q(v_r)
\right).
\]
Since \(P\), \(Q\), \(f\), and \(g\) are odd, the functions \(F_N\) and
\(G_N\) are odd. They therefore define an ordinary finite-dimensional
Krivine scheme.

Let \(H_N\) denote the correlation function of this scheme. Fix
\(t\in(-1,1)\), and let
\[
(U_1,V_1),\ldots,(U_N,V_N)
\]
be independent pairs of standard Gaussian vectors in \(\mathbb{R}^{3k}\)
satisfying
\[
\mathbb{E}[(U_r)_a(V_r)_b]
=
t\delta_{ab}.
\]
Define
\[
A_N
=
\frac{1}{\sqrt{N}}\sum_{r=1}^N P(U_r),
\qquad
B_N
=
\frac{1}{\sqrt{N}}\sum_{r=1}^N Q(V_r).
\]
By the definition of the correlation function,
\[
H_N(t)
=
\frac{\pi}{2}\,
\mathbb{E}\bigl[f(A_N)g(B_N)\bigr].
\]

\begin{claim}
As \(N \rightarrow \infty\), we have \[
(A_N,B_N)
\xrightarrow{d}
(X,Y),
\]
where \((X,Y)\) is jointly Gaussian and satisfies
\[
\mathbb{E}[X_iX_j]
=
\mathbb{E}[Y_iY_j]
=
\delta_{ij},
\qquad
\mathbb{E}[X_iY_j]
=
\delta_{ij}\rho_i(t).
\]
\end{claim}

\begin{claimproof}
The random vectors
\[
\bigl(P(U_r),Q(V_r)\bigr)\in\mathbb{R}^{2k}
\]
are independent, identically distributed, and centered. Because the
\(k\) Gaussian blocks are independent, their covariance matrix is
\[
\begin{pmatrix}
I_k & D(t)\\
D(t) & I_k
\end{pmatrix},
\qquad
D(t)
=
\operatorname{diag}\bigl(
\rho_1(t),\ldots,\rho_k(t)
\bigr).
\]
By the multivariate central limit theorem, we have
\(
(A_N,B_N)
\xrightarrow[N\to\infty]{d}
(X,Y),
\)
where \((X,Y)\) is jointly Gaussian satisfying
\(
\mathbb{E}[X_iX_j]
=
\mathbb{E}[Y_iY_j]
=
\delta_{ij}\) and 
\(
\mathbb{E}[X_iY_j]
=
\delta_{ij}\rho_i(t),
\)
as desired. 
\end{claimproof}

Returning to the proof of the theorem, since the function
\(
(x,y)\longmapsto f(x)g(y)
\)
 is bounded and has discontinuity set with measure 0, it follows that
\[
\mathbb{E}\bigl[f(X_N)g(Y_N)\bigr]
\longrightarrow
\mathbb{E}\bigl[f(X(t))g(Y(t))\bigr].
\]
Consequently,
\[
H_N(t)\longrightarrow H(t)
\]
for every \(t\in(-1,1)\), where
\[
H(t)
=
\frac{\pi}{2}\,
\mathbb{E}\bigl[f(X(t))g(Y(t))\bigr]
\]
is the correlation function of the limiting scheme.

It remains to convert this pointwise convergence to locally uniform
convergence on the unit disk. For each \(N\), the Hermite expansions of
\(F_N\) and \(G_N\) give a holomorphic extension of \(H_N\) to
\[
\mathbb{D}
=
\{z\in\mathbb{C}:|z|<1\}.
\]
More explicitly, if
\[
F_N=\sum_{\alpha}\widehat{F_N}(\alpha)\operatorname{He}_{\alpha},
\qquad
G_N=\sum_{\alpha}\widehat{G_N}(\alpha)\operatorname{He}_{\alpha},
\]
then
\[
H_N(z)
=
\frac{\pi}{2}
\sum_{\alpha}
\widehat{F_N}(\alpha)\widehat{G_N}(\alpha)z^{|\alpha|}.
\]
For \(z\) in the unit disk, the Cauchy--Schwarz inequality gives
\[
\begin{aligned}
|H_N(z)|
&\leq
\frac{\pi}{2}
\sum_{\alpha}
\bigl|
\widehat{F_N}(\alpha)\widehat{G_N}(\alpha)
\bigr|
|z|^{|\alpha|} \\
&\leq
\frac{\pi}{2}
\left(
\sum_{\alpha}
|\widehat{F_N}(\alpha)|^2
\right)^{1/2}
\left(
\sum_{\alpha}
|\widehat{G_N}(\alpha)|^2
\right)^{1/2} \\
&\leq
\frac{\pi}{2}
\|F_N\|_2\|G_N\|_2
=
\frac{\pi}{2}.
\end{aligned}
\]
Here the second inequality uses \(|z|<1\), and the third follows from
the orthonormality of the Hermite polynomials.
Thus, the family \((H_N)_{N\geq 1}\) is locally bounded on \(\mathbb{D}\).
Since it converges pointwise on the interval \((-1,1)\), which has an
accumulation point in \(\mathbb{D}\). Vitali's theorem implies that
\((H_N)\) converges locally uniformly on \(\mathbb{D}\) to a holomorphic
function. On \((-1,1)\), this function agrees with \(H\). It is therefore
the holomorphic extension of the correlation function of the limiting
scheme.
\end{proof}
We now show the following. 
\begin{theorem}[Upper bound from a limiting Krivine scheme]
\label{thm:upper_bound_limiting_scheme}
Let
\(
\mathcal{S}=(f,g,\rho_1,\ldots,\rho_k)
\)
be a limiting Krivine scheme. Let \(H\) be its correlation
function. Suppose that \(H'(0)\neq 0\), so that \(H\) has a local holomorphic inverse
at the origin. Write
\[
H^{-1}(z)
=
\sum_{j=0}^{\infty}a_{2j+1}z^{2j+1},
\]
and define its majorant by
\[
M(r)
=
\sum_{j=0}^{\infty}|a_{2j+1}|r^{2j+1}.
\]
If there exists \(\gamma>0\), within the radius of convergence of
\(H^{-1}\), such that
\(
M(\gamma)=1,
\)
then
\(
K_G
\leq
\frac{\pi}{2\gamma}.
\)
\end{theorem}

The only remaining thing left to show that an inverse-majorant bound for the
limiting correlation function passes to its finite-dimensional
approximations. We first prove this analytic stability statement.

\begin{lemma}[Finite inverse-majorant transfer]
\label{lem:finite-inverse-majorant-transfer}
Let $F_N,F$ be holomorphic on $\D$, suppose that $F_N\to F$ locally
uniformly on $\D$, and assume that $F_N(0)=F(0)=0$. Suppose that $G$ is a
holomorphic inverse branch of $F$ on a neighborhood of
$\overline{D(0,R_0)}$, takes values in $\D$, and satisfies
$
F(G(\zeta))=\zeta,$ with 
$
 |\zeta|\leq R_0.
$
Write
$
G(\zeta)=\sum_{n\geq 1}A_n\zeta^n.
$
If $0<\gamma<R_0$ and
$
\sum_{n\geq 1}|A_n|\gamma^n<1,
$
then, for all sufficiently large $N$, the function $F_N$ has a holomorphic
inverse branch
\[
G_N(\zeta)=F_N^{-1}(\zeta)=\sum_{n\geq 1}A_{n,N}\zeta^n
\]
defined on a neighborhood of $\overline{D(0,\gamma)}$, and
$
\sum_{n\geq 1}|A_{n,N}|\gamma^n<1.
$
Moreover, $A_{n,N}\to A_n$ for every fixed $n$.
\end{lemma}

\begin{proof}
The proof has three steps. First, Rouch\'e's theorem shows that the inverse
branch of $F$ persists for $F_N$. Second, the resulting inverse branches
converge to $G$, which controls every fixed coefficient. Finally, a Cauchy
estimate on a slightly larger disk controls the remaining tail uniformly in
$N$.

Choose radii
\(
\gamma<\rho<r<R<R_0
\)
and set
\(
\Omega=G(D(0,R)).
\)
Because $G$ is an inverse branch of $F$, it is injective and has nonzero
derivative on a neighborhood of $\overline{D(0,R)}$. Hence
\(
\partial\Omega=G(\partial D(0,R)).
\)
For $w\in\partial\Omega$, we therefore have $|F(w)|=R$. Thus, for every
$|\zeta|\leq r$,
\[
|F(w)-\zeta|\geq R-r,
\qquad w\in\partial\Omega.
\]
The contour $\partial\Omega$ is a compact subset of $\D$. Since
$F_N\to F$ locally uniformly, for all sufficiently large $N$,
\[
\sup_{w\in\partial\Omega}|F_N(w)-F(w)|<\frac{R-r}{2}.
\]
Fix such an $N$. For every $|\zeta|\leq r$, Rouch\'e's theorem implies that
$F_N(w)-\zeta$ and $F(w)-\zeta$ have the same number of zeros in $\Omega$,
counted with multiplicity. The equation $F(w)=\zeta$ has exactly one
solution in $\Omega$, namely $w=G(\zeta)$. Therefore $F_N(w)=\zeta$ also has
exactly one solution in $\Omega$.

These solutions define a holomorphic function $G_N$ on $D(0,r)$. One may
write it using the argument principle as
\[
G_N(\zeta)
=
\frac{1}{2\pi i}
\int_{\partial\Omega}
\frac{wF_N'(w)}{F_N(w)-\zeta}\,dw,
\qquad |\zeta|<r.
\]
This is the unique solution of $F_N(w)=\zeta$ in $\Omega$. In particular,
$G_N(0)=0$ and
\(
F_N(G_N(\zeta))=\zeta.
\)
Thus $G_N$ is the inverse branch of $F_N$ at the origin.

The same argument-principle formula holds for $G$. The denominators above
are uniformly bounded away from zero on
$\partial\Omega\times\overline{D(0,r)}$, and local uniform convergence of
$F_N$ implies uniform convergence of $F_N'$ to $F'$ on $\partial\Omega$.
It follows that
\(
G_N\longrightarrow G
\)
uniformly on $\overline{D(0,r)}$. By Cauchy's coefficient formula,
\(
A_{n,N}\longrightarrow A_n
\)
for every fixed $n$.

It remains to control all coefficients at once. Let
\(
B=\sup_{w\in\overline\Omega}|w|.
\)
Since $G_N(D(0,r))\subseteq\Omega$, Cauchy's estimate on the circle
$|\zeta|=\rho$ gives
\(
|A_{n,N}|\leq B\rho^{-n},
\qquad n\geq 1,
\)
for all sufficiently large $N$.

Set
\(
\delta
=
1-\sum_{n\geq 1}|A_n|\gamma^n
>0.
\)
Choose $M$ large enough that
\(
B\sum_{n>M}\left(\frac{\gamma}{\rho}\right)^n
<\frac{\delta}{3}.
\)
Since the first $M$ coefficients converge, for all sufficiently large $N$,
\(
\sum_{n\leq M}|A_{n,N}|\gamma^n
\leq
\sum_{n\leq M}|A_n|\gamma^n
+
\frac{\delta}{3}.
\)
Combining the head and tail bounds gives
\begin{align*}
\sum_{n\geq 1}|A_{n,N}|\gamma^n
&\leq
\sum_{n\leq M}|A_n|\gamma^n
+
\frac{\delta}{3}
+
B\sum_{n>M}\left(\frac{\gamma}{\rho}\right)^n \\
&<
\sum_{n\geq 1}|A_n|\gamma^n
+
\frac{2\delta}{3}
<1.
\end{align*}
This proves the claim.
\end{proof}
We will now prove ~\cref{thm:upper_bound_limiting_scheme}.

\begin{proof}
We first move from $\gamma$ to an arbitrary smaller radius $\gamma'<\gamma$, where the majorant inequality is strict. We then transfer this strict certificate to a finite-dimensional approximation and finally let $\gamma'$ approach $\gamma$.

Let $H_N$ be the correlation functions of the ordinary finite-dimensional
Krivine schemes supplied by the finite-dimensional approximation theorem.
Thus
\(
H_N\longrightarrow H
\)
locally uniformly on $\D$. Since all of the schemes are odd,
$H_N(0)=H(0)=0$.

The equality $M(\gamma)=1$ need not remain true under a small perturbation,
so we first work at a slightly smaller radius. Fix any
\(
0<\gamma'<\gamma.
\)
Because
\(
a_1=\frac{1}{H'(0)}\neq 0,
\)
the majorant is strictly increasing on the positive real axis. Hence
\(
M(\gamma')<M(\gamma)=1.
\)
Choose $R_0$ such that
\(
\gamma'<R_0<\gamma.
\)
Then $M(R_0)<1$. If $G=H^{-1}$ denotes the inverse power series, then for
$|z|\leq R_0$,
\(
|G(z)|
\leq
M(|z|)
\leq
M(R_0)
<1.
\)
Thus $G$ is holomorphic on a neighborhood of
$\overline{D(0,R_0)}$ and takes values in $\D$. Moreover, the identity
$H(G(z))=z$ holds near the origin. Both sides are holomorphic on
$D(0,R_0)$, so the identity theorem gives
\[
H(G(z))=z,
\qquad |z|<R_0.
\]
Therefore $G$ is an inverse branch of $H$ on this disk.

We may now apply Lemma~\ref{lem:finite-inverse-majorant-transfer} at the
radius $\gamma'$. For all sufficiently large $N$, the inverse branch of
$H_N$ has an expansion
\[
H_N^{-1}(z)=\sum_{n\geq 1}a_{n,N}z^n
\]
satisfying
\[
\sum_{n\geq 1}|a_{n,N}|(\gamma')^n<1.
\]
The classical Krivine criterion applied to this finite-dimensional scheme
therefore yields
\(
\KG\leq \frac{\pi}{2\gamma'}.
\)
This holds for every $\gamma'<\gamma$. Letting $\gamma'$ increase to
$\gamma$ gives
\[
\KG\leq \frac{\pi}{2\gamma},
\]
as claimed.
\end{proof}

%% file: paper/asymptotic_scheme.tex
\section{The Cubic--Quintic Scheme}
\label{sec:scheme}

The previous section lets us work directly with fixed-dimensional limiting
Krivine schemes. We now define the two-dimensional scheme used for our upper
bound. Its first coordinate combines the correlation powers \(t,t^3,t^5\),
while its second coordinate keeps correlation \(t\). In the next two sections,
we compute its correlation function and certify its inverse majorant.

We first describe the general form of the construction. Let
\(D\subseteq\{3,5,7,\ldots\}\) be finite, choose weights
\((s_d)_{d\in D}\), and set
\[
    \sigma_d=(-1)^{(d-1)/2},
    \qquad
    V_D=1+\sum_{d\in D}s_d^2.
\]
A degree-\(d\) Hermite direction contributes the correlation \(t^d\).
After combining these directions into one normalized Gaussian coordinate,
the resulting correlation map is
\begin{equation}
    \rho_D(t)
    =
    \frac{t+\sum_{d\in D}\sigma_d s_d^2 t^d}{V_D}.
    \label{eq:general-hermite-correlation-map}
\end{equation}
The map \(\rho_D\) is allowable, since the absolute values of its
coefficients sum to one. Thus, for any \(\eta\in\mathbb{R}\), the partitions
\[
    \operatorname{sgn}\left(
        w+\frac{\eta}{\sqrt{V_D}}\He_3(x)
    \right)
    \qquad\text{and}\qquad
    \operatorname{sgn}\left(
        w-\frac{\eta}{\sqrt{V_D}}\He_3(x)
    \right)
\]
together with the correlation maps \(\rho_D(t)\) and \(t\) define a
two-dimensional limiting Krivine scheme.

For the scheme used in our proof, we take \(D=\{3,5\}\). The cubic
direction is anti-aligned between the two sides, while the quintic direction
is aligned. We use the parameters
\begin{equation}
    \eta=0.136419125,
    \qquad
    s_3=0.34101124,
    \qquad
    s_5=0.05276111.
    \label{eq:params}
\end{equation}
Set
\[
    V=1+s_3^2+s_5^2,
    \qquad
    \vartheta=\frac{\eta}{\sqrt V},
\]
and define
\begin{equation}
    \rho(t)
    =
    \frac{t-s_3^2t^3+s_5^2t^5}{V}.
    \label{eq:rho-collapsed}
\end{equation}

\begin{definition}[Cubic--quintic limiting scheme]
\label{def:cubic-quintic-scheme}
Define \(f,g:\mathbb{R}^2\to\{-1,1\}\) by
\begin{equation}
    f(w,x)
    =
    \sgn\bigl(w+\vartheta\He_3(x)\bigr),
    \qquad
    g(w,x)
    =
    \sgn\bigl(w-\vartheta\He_3(x)\bigr).
    \label{eq:cubic-quintic-partitions}
\end{equation}
Our construction is the limiting Krivine scheme
\[
    \mathcal{S}=(f,g,\rho_1,\rho_2),
    \qquad
    \rho_1(t)=\rho(t),
    \qquad
    \rho_2(t)=t.
\]
We write \(H\) for its correlation function.
\end{definition}

The map \(\rho\) is allowable because the absolute values of its
coefficients sum to \((1+s_3^2+s_5^2)/V=1\). The partitions are odd and
satisfy \(g(w,x)=f(w,-x)\). Their threshold boundaries also have Gaussian
measure zero, so the approximation theorem from the previous section
applies.

More explicitly, let \((W,X)\) and \((W',X')\) be standard Gaussian vectors.
Assume that the two coordinate pairs are independent and that
\(\operatorname{corr}(W,W')=\rho(t)\) and
\(\operatorname{corr}(X,X')=t\). Then
\begin{equation}
    H(t)
    =
    \frac{\pi}{2}
    \E\bigl[f(W,X)g(W',X')\bigr].
    \label{eq:lim-H}
\end{equation}

The signs and weights in \(\rho\) are chosen to make the cubic and quintic
coefficients of \(H\) nearly vanish. This in turn suppresses the first two
nonlinear coefficients of \(H^{-1}\). Section~\ref{sec:coefficients} makes
this relation explicit.

\begin{theorem}[Cubic--quintic upper bound]
\label{thm:cubic-quintic-upper-bound}
Let \(M\) be the inverse majorant of \(H\). At
\(\gamma=0.881545409\), we have \(M(\gamma)<1\). Consequently,
\[
    K_G
    \leq
    \frac{\pi}{2\gamma}
    \leq
    1.7818666069360661
    <
    \frac{\pi}{2\log(1+\sqrt{2})}
    -
    3.4737\cdot 10^{-4}.
\]
\end{theorem}

%% file: paper/coefficient_formula.tex

\section{The Correlation Coefficients}
\label{sec:coefficients}

The previous section defined our two-dimensional limiting scheme, whose two
correlation maps are \(\rho(t)\) and \(t\). We now compute its correlation
function
\[
  H(t)=\sum_{m\geq 1} b_m t^m.
\]
The key observation is that \(g(w,x)=f(w,-x)\). Thus, the Hermite
coefficients of \(f\) determine the entire series \(H\).

We use standard facts about Hermite polynomials. See  \cite{Andrews_Askey_Roy_1999} for a reference.  

Let $\mu$ be the standard Gaussian measure on $\R$, so that
\[
    d\mu(x):=\frac{1}{\sqrt{2\pi}}e^{-x^2/2}\,dx.
\]
Let \(\HermBasis_n\) denote the \(n\)-th orthonormal probabilist's Hermite polynomial. Then \(\{\HermBasis_n\}_{n\ge0}\) is an orthonormal basis for \(L^2(\mu)\).

Let \(W,X\) be independent standard Gaussians. For \(a,b\geq 0\), write
\begin{equation}
  \mathsf A_{a,b}
  =
  \E\!\left[
    f(W,X)\He_a(W)\He_b(X)
  \right].
  \label{eq:mathsfAab}
\end{equation}
Then, in \(L_2\) of the standard Gaussian measure on \(\mathbb{R}^2\),
\[
  f(w,x)
  =
  \sum_{a,b\geq 0}
  \mathsf A_{a,b}\He_a(w)\He_b(x).
\]
Since \(f(-w,-x)=-f(w,x)\), we have \(\mathsf A_{a,b}=0\) whenever
\(a+b\) is even. For a power series \(P\), we write \([t^m]P(t)\) for the
coefficient of \(t^m\).

\begin{theorem}[Coefficient formula]
\label{thm:coefficient-formula}
For every \(t\in(-1,1)\),
\begin{equation}
  H(t)
  =
  \frac{\pi}{2}
  \sum_{a,b\geq 0}
  \mathsf A_{a,b}^2\rho(t)^a(-t)^b.
  \label{eq:H-coefficient-formula}
\end{equation}
Consequently, for every \(m\geq 1\),
\begin{equation}
  b_m
  =
  \frac{\pi}{2}
  \sum_{\substack{a,b\geq 0\\a+b\leq m}}
  (-1)^b\mathsf A_{a,b}^2
  [t^{m-b}]\rho(t)^a.
  \label{eq:bm-coefficient-formula}
\end{equation}
In particular, \(b_m=0\) whenever \(m\) is even.
\end{theorem}

\begin{proof}
Let \((W,X)\) and \((W',X')\) be the Gaussian inputs used to define the
correlation function of the scheme. Thus, \(W\) and \(W'\) have correlation
\(\rho(t)\), while \(X\) and \(X'\) have correlation \(t\); the two
coordinate pairs are independent.

Since \(g(w,x)=f(w,-x)\), the Hermite expansion of \(g\) is
\[
  g(w,x)
  =
  \sum_{a,b\geq 0}
  (-1)^b\mathsf A_{a,b}\He_a(w)\He_b(x).
\]
When the two expansions are substituted into the definition of \(H\),
orthogonality leaves only matching Hermite degrees. The Hermite covariance
identity then contributes \(\rho(t)^a\) in the first coordinate and \(t^b\)
in the second. This gives \eqref{eq:H-coefficient-formula}. The series is
absolutely convergent for \(t\in(-1,1)\), since \(f\in L_2\) and
\(|\rho(t)|\leq |t|<1\).

Extracting the coefficient of \(t^m\) gives
\eqref{eq:bm-coefficient-formula}. Finally, \(\rho\) is odd and
\(\mathsf A_{a,b}\) vanishes unless \(a+b\) is odd. Hence every nonzero
summand in \eqref{eq:H-coefficient-formula} is odd in \(t\), so only odd
coefficients occur.
\end{proof}

\paragraph{Computing the Hermite coefficients.}
The coefficients in \eqref{eq:mathsfAab} reduce to
one-dimensional Gaussian integrals. Let \(Z\) be a standard Gaussian and,
for \(a\geq 0\), set
\begin{equation}
  q_a(u)
  =
  \E\!\left[\sgn(Z+u)\He_a(Z)\right].
  \label{eq:threshold-hermite-coefficients}
\end{equation}
Conditioning on \(X\) gives
\begin{equation}
  \mathsf A_{a,b}
  =
  \E\!\left[
    \He_b(X)q_a\!\left(\vartheta\He_3(X)\right)
  \right].
  \label{eq:f-coefficients-one-dimensional}
\end{equation}
The threshold coefficients have the closed form
\begin{equation}
  q_0(u)
  =
  \erf\!\left(\frac{u}{\sqrt{2}}\right),
  \qquad
  q_a(u)
  =
  \frac{2\varphi(u)}{\sqrt{a}}\He_{a-1}(-u)
  \quad (a\geq 1),
  \label{eq:threshold-hermite-closed-form}
\end{equation}
where \(\varphi(u)=(2\pi)^{-1/2}e^{-u^2/2}\). This standard identity follows
by one integration by parts; we include the proof in
Appendix~\ref{app:coefficient-formulas}.

Theorem~\ref{thm:coefficient-formula} therefore reduces each \(b_m\) to
finite coefficient extraction from powers of \(\rho\) and one-dimensional
Gaussian integrals. This is the formula used to compute the coefficient head
in the next section. A more general formula for schemes built from an
arbitrary finite set of Hermite degrees is recorded in
Appendix~\ref{app:coefficient-formulas}; it is not needed for the main
argument.

%% file: paper/certification.tex
\section{Proof of Cubic-Quintic Scheme}
\label{sec:certification}

Write the correlation function of the cubic--quintic limiting scheme as
\begin{equation}
  H(t)=b_1t+\sum_{\substack{m\ge 3\\ m\ \mathrm{odd}}} b_m t^m,
  \label{eq:certification-H-series}
\end{equation}
and let
\begin{equation}
  \Delta_H:=\sum_{\substack{m\ge 3\\ m\ \mathrm{odd}}}|b_m|
  \label{eq:nonlinear-coefficient-mass}
\end{equation}
be the total mass of its nonlinear coefficients.  Let
\[
  G(z)=H^{-1}(z)=\sum_{n\ge1}a_nz^n
\]
denote the local holomorphic inverse of $H$ at the origin, and define its
majorant by
\begin{equation}
  \mathcal M(r):=\sum_{n\ge1}|a_n|r^n.
  \label{eq:inverse-majorant-certification}
\end{equation}
The local inverse exists because $H'(0)=b_1>0$, by the holomorphic inverse
function theorem.

For the linear map $t\mapsto b_1t$, the condition $\mathcal M(\gamma)<1$ is
simply $\gamma<b_1$.  The next proposition shows that the nonlinear part costs
at most $\Delta_H$ of this margin.

\begin{proposition}[Inverse certificate from near-linearity]
\label{prop:inverse-certificate-near-linearity}
Assume $b_1>0$ and $\Delta_H<\infty$.  If
\begin{equation}
  \gamma+\Delta_H<b_1,
  \label{eq:near-linearity-condition}
\end{equation}
then there exists $R>\gamma$ such that the local inverse $G$ extends
holomorphically to $R\mathbb D$, maps $R\mathbb D$ into $\mathbb D$, and
satisfies
\begin{equation}
  \mathcal M(\gamma)
  \le \frac{\gamma+\Delta_H}{b_1}
  <1.
  \label{eq:near-linearity-conclusion}
\end{equation}
\end{proposition}

\begin{proof}
Set
\[
  E(t):=\sum_{m\ge3}b_mt^m,
  \qquad
  E_{\mathrm{maj}}(u):=\sum_{m\ge3}|b_m|u^m.
\]
There is a unique formal power series
\[
  V(z)=\sum_{n\ge1}v_nz^n
\]
satisfying
\begin{equation}
  b_1V(z)=z+E_{\mathrm{maj}}(V(z)).
  \label{eq:majorizing-inverse-equation}
\end{equation}
Indeed, comparison of the coefficient of $z$ gives $v_1=b_1^{-1}$.  For
$n\ge2$, the coefficient of $z^n$ in $E_{\mathrm{maj}}(V)$ depends only on
$v_1,\ldots,v_{n-1}$, so \eqref{eq:majorizing-inverse-equation} determines
$v_n$ recursively.  The same recursion shows that $v_n\ge0$ for every $n$.

Comparing coefficients in
\[
  b_1G(z)=z-E(G(z))
\]
and in \eqref{eq:majorizing-inverse-equation}, induction on $n$ gives
\begin{equation}
  |a_n|\le v_n
  \qquad(n\ge1).
  \label{eq:inverse-coefficient-domination}
\end{equation}
Indeed, at degree $n$, every nonlinear composition involves only previously
determined coefficients, and replacing each $b_m$ by $|b_m|$ removes all
possible cancellation.

Choose $R>\gamma$ so close to $\gamma$ that
\[
  q:=\frac{R+\Delta_H}{b_1}<1.
\]
Starting from $V_0(z)=0$, define
\begin{equation}
  V_{k+1}(z):=\frac{z+E_{\mathrm{maj}}(V_k(z))}{b_1}.
  \label{eq:majorizing-inverse-iteration}
\end{equation}
The series $V_k$ have nonnegative coefficients and increase coefficientwise,
meaning that each Taylor coefficient is nondecreasing in $k$.
Moreover, if $V_k(R)\le q$, then $V_k(R)<1$ and
\[
  V_{k+1}(R)
  \le \frac{R+E_{\mathrm{maj}}(1)}{b_1}
  =q.
\]
Thus $V_k(R)\le q$ for every $k$.  Their coefficientwise limit satisfies
\eqref{eq:majorizing-inverse-equation}, and hence equals $V$ by uniqueness.
The monotone convergence theorem now gives
\[
  V(R)=\sum_{n\ge1}v_nR^n\le q<1.
\]
Together with \eqref{eq:inverse-coefficient-domination}, this implies
\[
  \sum_{n\ge1}|a_n|R^n\le V(R)<1.
\]
Consequently, the Taylor series of $G$ converges on $R\mathbb D$ and takes
values in $\mathbb D$.  Since the identity $H(G(z))=z$ holds near the origin,
it holds throughout $R\mathbb D$ by the identity theorem.  Finally,
\[
  \mathcal M(\gamma)
  \le V(\gamma)
  =\frac{\gamma+E_{\mathrm{maj}}(V(\gamma))}{b_1}
  \le\frac{\gamma+\Delta_H}{b_1},
\]
because $0\le V(\gamma)\le V(R)<1$.
\end{proof}

It remains to certify $\Delta_H$.  We compute the coefficients through a
finite degree $N$ and bound the remaining tail indirectly.  The point is that
the Euler operator
\[
  D:=t\frac{d}{dt}
\]
multiplies the coefficient of $t^m$ by $m$, so $D^3$ strongly amplifies
high-degree coefficients.  The $L^2(\mathbb T)$ norm below is the
root-mean-square size of the resulting function around the unit circle
$\mathbb T:=\{z\in\mathbb C:|z|=1\}$.

\begin{lemma}[Third-order coefficient-tail bound]
\label{lem:third-order-tail}
Suppose that $D^3H$ has radial boundary values in $L^2(\mathbb T)$, with
\[
  \|D^3H\|_{L^2(\mathbb T)}^2
  :=\frac{1}{2\pi}\int_0^{2\pi}
  |D^3H(e^{i\theta})|^2\,d\theta.
\]
Then, for every odd $N$,
\begin{equation}
  \sum_{\substack{m>N\\m\ \mathrm{odd}}}|b_m|
  \le
  \|D^3H\|_{L^2(\mathbb T)}
  \left(\sum_{\substack{m>N\\m\ \mathrm{odd}}}m^{-6}\right)^{1/2}
  \le
  \frac{\|D^3H\|_{L^2(\mathbb T)}}{\sqrt{10N^5}}.
  \label{eq:third-order-tail}
\end{equation}
\end{lemma}

\begin{proof}
Since $D^3(t^m)=m^3t^m$, Parseval's identity gives
\begin{equation}
  \|D^3H\|_{L^2(\mathbb T)}^2
  =\sum_{\substack{m\ge1\\m\ \mathrm{odd}}}m^6|b_m|^2.
  \label{eq:parseval-third-order}
\end{equation}
The Cauchy--Schwarz inequality therefore gives
\[
  \sum_{\substack{m>N\\m\ \mathrm{odd}}}|b_m|
  =\sum_{\substack{m>N\\m\ \mathrm{odd}}}
  m^{-3}\bigl(m^3|b_m|\bigr)
  \le
  \left(\sum_{\substack{m>N\\m\ \mathrm{odd}}}m^{-6}\right)^{1/2}
  \|D^3H\|_{L^2(\mathbb T)}.
\]
Finally, since $x\mapsto x^{-6}$ is decreasing and consecutive odd integers
are separated by two, the monotone integral estimate gives
\[
  \sum_{\substack{m>N\\m\ \mathrm{odd}}}m^{-6}
  \le\frac12\int_N^\infty x^{-6}\,dx
  =\frac{1}{10N^5}.
\]
\end{proof}

\begin{proposition}[Certified numerical input]
\label{prop:certified-numerical-input}
For the cubic--quintic correlation function, interval arithmetic certifies
\begin{equation}
  b_1\ge0.881573822049,
  \qquad
  \sum_{\substack{3\le m\le251\\m\ \mathrm{odd}}}|b_m|
  \le1.1328860\cdot10^{-5},
  \qquad
  \|D^3H\|_{L^2(\mathbb T)}<14.443.
  \label{eq:certified-numerical-input}
\end{equation}
\end{proposition}

\begin{proof}
See Appendix~\ref{app:certification-details}.
\end{proof}

Lemma~\ref{lem:third-order-tail} and
Proposition~\ref{prop:certified-numerical-input} give
\begin{equation}
  \Delta_H
  \le
  1.1328860\cdot10^{-5}
  +\frac{14.443}{\sqrt{10\cdot251^5}}
  <1.5904724\cdot10^{-5}.
  \label{eq:certified-nonlinear-mass}
\end{equation}

\begin{theorem}[Cubic--quintic upper bound]
\label{thm:cubic-quintic-upper-bound-certified}
For
\[
  \gamma=0.881545409,
\]
the cubic--quintic limiting scheme satisfies $\mathcal M(\gamma)<1$.
Consequently,
\begin{equation}
  K_G\le\frac{\pi}{2\gamma}
  \le1.7818666069360661.
  \label{eq:certified-KG-bound}
\end{equation}
\end{theorem}

\begin{proof}
By \eqref{eq:certified-nonlinear-mass},
\[
  \gamma+\Delta_H
  <0.881545409+1.5904724\cdot10^{-5}
   =0.881561313724
   <0.881573822049
   \le b_1
\]
Proposition~\ref{prop:inverse-certificate-near-linearity} therefore gives
$\mathcal M(\gamma)<1$, together with the analytic room required by the
finite-dimensional transfer theorem.  That theorem passes the strict
certificate to all sufficiently large ordinary finite-dimensional
approximations.  The classical Krivine preprocessing theorem then gives
\eqref{eq:certified-KG-bound}.
\end{proof}

\paragraph{Reproducibility.}
  The scripts, certificates, and Arb outputs
  are available at \url{https://github.com/trishullab/grothendieck-bounds}.

%% file: paper/lower_bound_general.tex
\section{Proof Overview}
\label{sec:proof-idea}


\begin{theorem}\label{thm:lower-bound}
The Grothendieck constant satisfies
\[
        K_G\ge\frac{6\pi}{11} = 1.7136\ldots.
\]
\end{theorem}

A sketch of the proof is as follows. First, recall the inverse-series step in Krivine's rounding argument. If
\[
        H^{-1}(z)=\sum_{j\ge0}a_{2j+1}z^{2j+1},
\]
define the absolute-coefficient majorant $M_H(s):=\sum_{j\ge0}|a_{2j+1}|s^{2j+1}$. The standard tensor preprocessing gives the implication
\begin{equation}
        M_H(\gamma)\le1
        \quad\Longrightarrow\quad
        K_G\le\frac{\pi}{2\gamma}.
        \label{eq:inverse-majorant-bound}
\end{equation}
We refer to~\eqref{eq:inverse-majorant-bound} as the \emph{inverse-majorant bound}.

The conceptual starting point is the theorem of Naor and Regev that Krivine schemes are optimal \cite{naor2014krivine}. By reworking the construction in their proof, we show that the inverse-majorant bound is itself asymptotically optimal after allowing mixtures and increasing dimension: there are limiting mixed correlation functions with admissible parameters converging to $\pi/(2K_G)$. This transfer is established in Section~\ref{sec:majorant_optimal}.

It therefore remains to prove a universal upper bound on such admissible parameters. For arbitrary odd measurable sign functions $f,g:\R^d\to\{\pm1\}$, write
\[
        H_{f,g}(t)=b_1t+b_3t^3+b_5t^5+\cdots.
\]
The technical core of the argument is the affine coefficient inequality
\[
        b_3\ge 2b_1-\frac{11}{6}.
\]
Section~\ref{sec:preliminaries} sets up the Gaussian--Hermite formulation, and Section~\ref{sec:affine-strip} proves this inequality by decomposing $f$ and $g$ into their agreement and disagreement parts and reducing the required estimates to one-dimensional inequalities. Its linearity in $(b_1,b_3)$ is essential, because it makes the inequality stable under mixtures and coefficientwise limits.

Section~\ref{sec:strip_majorant} then shows that the affine coefficient inequality controls the inverse-majorant bound: every admissible value must satisfy $\gamma\le11/12$. Applying this barrier to the asymptotically optimal sequence constructed in Section~\ref{sec:majorant_optimal} gives
\[
        \frac{\pi}{2K_G}\le\frac{11}{12},
\]
which is exactly Theorem~\ref{thm:lower-bound}.

\section{Preliminaries on Gaussian Hermite Expansions}
\label{sec:preliminaries}

Let
\[
        X\sim N(0,I_d)
\]
be a standard Gaussian vector in $\R^d$. We define the \textit{inner product} as the Gaussian expectation:
\[
        \ip{F}{G}=\E[F(X)G(X)].
\]
We use the orthonormal 1D Hermites
\[
        \HermBasis_0=1,\qquad
        \HermBasis_1(s)=s,\qquad
        \HermBasis_2(s)=\frac{s^2-1}{\sqrt2},\qquad
        \HermBasis_3(s)=\frac{s^3-3s}{\sqrt6}.
\]

In $d$ dimensions the Hermite basis is indexed by multi-indices. If
\[
        \alpha=(\alpha_1,\ldots,\alpha_d),
        \qquad |\alpha|=\alpha_1+\cdots+\alpha_d,
\]
then
\[
        \HermBasis_\alpha(X)=\prod_{i=1}^d \HermBasis_{\alpha_i}(X_i).
\]
Any square integrable function has a Hermite expansion, written as
\[
        F=P_0F+P_1F+P_2F+P_3F+\ldots,
\]
where $P_iF$ is the orthogonal projection of $F$ onto the span of all $\HermBasis_\alpha$ with total degree $|\alpha|=i$.

In particular, the degree-one Hermites span
\[
        \cH_1=\operatorname{span}\{X_1,\ldots,X_d\}.
\]
Here,
\[
        P_1F=\ip{c_F}{X},
        \qquad
        c_F=\E[F(X)X]\in\R^d.
\]
Its norm is
\[
        \norm{P_1F}^2=\norm{c_F}_{\ell_2}^2.
\]

Similarly, $P_3F$ is the projection onto the full total degree-three Hermite span. We will not list all basis elements. The important point is that it includes pure cubic directions such as $\HermBasis_3(X_i)$ and mixed directions such as $\HermBasis_2(X_i)\HermBasis_1(X_j)$ and $\HermBasis_1(X_i)\HermBasis_1(X_j)\HermBasis_1(X_\ell)$.

For signs $f,g:\R^d\to\{\pm1\}$, if
\[
        H_{f,g}(t)=\frac\pi2\E[f(X)g(X'_t)]
\]
with $(X,X'_t)$ a pair of correlated standard Gaussian vectors, then Mehler's identity gives
\[
        H_{f,g}(t)=\frac\pi2\sum_{m\ge0}\ip{P_mf}{P_mg}t^m.
\]
Here
\[
        b_1=[t]H_{f,g}=\frac\pi2\ip{P_1f}{P_1g},
        \qquad
        b_3=[t^3]H_{f,g}=\frac\pi2\ip{P_3f}{P_3g}.
\]

For the rest of the paper, we will use the shorthand
\[
        \nu:=\sqrt{\frac{2}{\pi}}=\E|Z|,
        \qquad Z\sim N(0,1).
\]

We want to prove
\begin{equation}
        b_3\ge 2b_1-\frac{11}{6}.
        \label{eq:affine-strip}
\end{equation}

For example, the hyperplane partitions lie on this line. For $f=g=\sgn(X_1)$,
\[
        \norm{P_1f}=\nu,
        \qquad
        \norm{P_3f}^2=\frac{\nu^2}{6}.
\]
Thus
\[
        (b_1,b_3)=\left(1,\frac16\right),
\]
and this point lies exactly on the line since $1/6=2-11/6$.

Instead of proving the strip directly, it suffices to prove
\begin{equation}
        B_1(f,g)
        :=\norm{P_1f}\norm{P_1g}
          +\ip{P_1f}{P_1g}
          -\ip{P_3f}{P_3g}
        \le \frac{11}{6}\nu^2.
        \label{eq:b1-bound}
\end{equation}
Indeed, by Cauchy--Schwarz,
\[
        \norm{P_1f}\norm{P_1g}\ge \ip{P_1f}{P_1g}.
\]
So~\eqref{eq:b1-bound} implies
\[
        2\ip{P_1f}{P_1g}-\ip{P_3f}{P_3g}\le \frac{11}{6}\nu^2.
\]
Multiplying by $\pi/2$ and using $\nu^2=2/\pi$ gives~\eqref{eq:affine-strip}.

\section{Proof of the Affine Coefficient Inequality}
\label{sec:affine-strip}

\subsection{The agreement--disagreement substitution}
\label{sec:agreement-disagreement}

We define a change of variables based on where $f$ and $g$ agree or disagree. Set the functions $h,k$ to be
\[
        h=\frac{f+g}{2},\qquad k=\frac{f-g}{2}.
\]
Then
\[
        f=h+k,
        \qquad
        g=h-k.
\]
Since $f,g$ are signs,
\[
        h,k\in\{-1,0,1\},
        \qquad
        hk=0,
        \qquad
        h^2+k^2=1.
\]
Think of $h$ as the agreement part and $k$ as the disagreement part between $f$ and $g$. This change of variables considerably simplifies the proof.

Let
\[
        u=P_1h,
        \qquad
        v=P_1k.
\]
Then
\[
        P_1f=u+v,
        \qquad
        P_1g=u-v.
\]
Hence
\[
        \ip{P_1f}{P_1g}=\norm{u}^2-\norm{v}^2.
\]
Also
\[
        \norm{P_1f}\norm{P_1g}
        =\norm{u+v}\norm{u-v}
        \le \norm{u}^2+\norm{v}^2,
\]
using $ab\le(a^2+b^2)/2$ after expanding $\norm{u\pm v}^2$.

For degree three,
\[
        P_3f=P_3h+P_3k,
        \qquad
        P_3g=P_3h-P_3k,
\]
so
\[
        \ip{P_3f}{P_3g}=\norm{P_3h}^2-\norm{P_3k}^2.
\]
Putting these together,
\begin{align*}
        B_1(f,g)
        &\le
        \bigl(\norm{u}^2+\norm{v}^2\bigr)
        +\bigl(\norm{u}^2-\norm{v}^2\bigr)
        -\bigl(\norm{P_3h}^2-\norm{P_3k}^2\bigr)\\
        &=2\norm{P_1h}^2-\norm{P_3h}^2+\norm{P_3k}^2.
\end{align*}

So it is enough to prove
\begin{equation}
\boxed{\displaystyle
        D_1(h,k):=2\norm{P_1h}^2-\norm{P_3h}^2+\norm{P_3k}^2
        \le \frac{11}{6}\nu^2
}
\label{eq:hk-reduction}
\end{equation}

In order to prove this, we have to handle each of the three expressions $\norm{P_1h}^2$, $\norm{P_3h}^2$, $\norm{P_3k}^2$ carefully. In particular, we need to prove upper bounds on $\norm{P_1h}^2$, $\norm{P_3k}^2$ and a lower bound on $\norm{P_3h}^2$.

\begin{remark}
    Note that in the inequality before~\eqref{eq:hk-reduction} the term $\norm{v}^2=\norm{P_1k}^2$ conveniently cancels. This is a result of the choice of the strip $b_3\ge 2b_1-11/6$. Stronger lower bounds can be proved with the same method, but this is more complicated because one also has to bound $\norm{P_1k}^2$.
\end{remark}

\subsection{Reduction to one Gaussian coordinate}
\label{sec:rotation}

The linear part of $h$ has the form
\[
        P_1h(X)=\ip{c_h}{X}.
\]
The coefficient vector is $c_h\in\R^d$, and
\[
        \norm{P_1h}=\norm{c_h}_{\ell_2}.
\]
If $c_h\neq0$, rotate coordinates so the first coordinate points in the $c_h$ direction. After this rotation, and after renaming the coordinates, we write
\[
        X=(S,T),
        \qquad S\sim N(0,1),\quad T\sim N(0,I_{d-1}),
\]
with $S$ independent of $T$, and
\[
        P_1h(S,T)=\norm{P_1h}\,S.
\]
If $P_1h=0$, choose any direction as $S$.

Define
\[
        x:=\frac{\norm{P_1h}}{\nu}\in[0,1].
\]
The bound $x\le1$ follows from
\[
        \norm{P_1h}=\sup_{\norm{e}=1}\E[h(X)\ip{e}{X}]
        \le \E|Z|=\nu.
\]
After rotation,
\[
        \ip{h}{S}=\norm{P_1h}=\nu x.
\]

\subsection{The agreement terms}
\label{sec:agreement}

The agreement part of $D_1$ is $2\norm{P_1h}^2-\norm{P_3h}^2$. The helpful negative term is
\[
        -\norm{P_3h}^2.
\]
To use it, we only need a lower bound on $\norm{P_3h}^2$.

In any dimension,
\[
        \norm{P_3h}^2
        \ge \ip{h}{\HermBasis_3(S)}^2,
\]
because $\HermBasis_3(S)$ is one orthonormal degree-three Hermite direction. The trick we will use is to condition on the transverse variables, so define
\[
        A(s):=\E_T[h(s,T)].
\]
This is the average value of $h$ along the slice $S=s$. Since $h\in[-1,1]$,
\[
        |A(s)|\le1.
\]
Also, for every function $\varphi$ depending only on $S$, it is a straightforward conditioning to show that
\[
        \E[h(S,T)\varphi(S)]=\E[A(S)\varphi(S)].
\]
In particular,
\[
        \E[A(S)S]=\ip{h}{S}=\nu x
\]
and
\[
        \ip{h}{\HermBasis_3(S)}=\E[A(S)\HermBasis_3(S)].
\]

This is the same trick as in the two-dimensional note. The problem is now a variational problem in the one variable function $A$, even though we are working in $d$ dimensions. To bound the agreement part, it is sufficient to relax $A$, and answer the following question:

\begin{quote}
``Among all $A:\R\to[-1,1]$ with $\E[A(S)S]=\nu x$, how small can the cubic coefficient be?''
\end{quote}

The answer is the rearrangement lemma.

Define
\[
        W(S):=S^3-3S=\sqrt6\,\HermBasis_3(S).
\]
Also define
\[
        \Delta(x):=x+2(1+x)\log\frac{1+x}{2},
        \qquad
        \Delta_+(x):=\max\{\Delta(x),0\}.
\]

\begin{lemma}[One-dimensional rearrangement]
Let $S\sim N(0,1)$ and $x\in[0,1]$. Then
\[
        \sup_{\substack{|A|\le1\\ \E[A(S)S]=\nu x}}
        \E[A(S)(S^3-3S)]
        =-\nu\Delta(x).
\]
\end{lemma}

\begin{proof}
Fix $r\ge0$ by
\[
        e^{-r^2/2}=\frac{1+x}{2},
\]
and put
\[
        \mu=r^2-3.
\]
Let $W=S^3-3S$. Then
\[
        W-\mu S=S^3-3S-(r^2-3)S=S(S^2-r^2).
\]

For any feasible $A$,
\begin{align*}
        \E[AW]
        &=\E[A(W-\mu S)]+\mu\E[AS]\\
        &=\E[A(W-\mu S)]+\mu\nu x\\
        &\le \E|W-\mu S|+\mu\nu x,
\end{align*}
using $|A|\le1$.

Equality is attained by
\[
        A_*(S)=\sgn(W-\mu S)=\sgn(S)\sgn(S^2-r^2),
\]
provided $A_*$ has the right first moment. It does, because
\[
        A_*(S)S=|S|\sgn(S^2-r^2),
\]
so
\begin{align*}
        \E[A_*(S)S]
        &=2\E[|S|\one_{|S|>r}]-\E|S|\\
        &=2\nu e^{-r^2/2}-\nu\\
        &=\nu x.
\end{align*}
Thus $A_*$ is feasible and attains the dual upper bound.

It remains to compute its value. Since
\[
        \sgn(S)(S^3-3S)=|S|^3-3|S|,
\]
we get
\begin{align*}
        \E[A_*W]
        &=\E[(|S|^3-3|S|)\sgn(S^2-r^2)]\\
        &=2\E[(|S|^3-3|S|)\one_{|S|>r}]
          -\E[|S|^3-3|S|].
\end{align*}
The Gaussian tail moments are
\[
        \E[|S|\one_{|S|>r}]=\nu e^{-r^2/2},
        \qquad
        \E[|S|^3\one_{|S|>r}]=\nu(r^2+2)e^{-r^2/2}.
\]
Also
\[
        \E|S|=\nu,
        \qquad
        \E|S|^3=2\nu.
\]
Therefore
\begin{align*}
        \E[A_*W]
        &=2\nu(r^2-1)e^{-r^2/2}+\nu\\
        &=\nu\left(1+(1+x)(r^2-1)\right)\\
        &=\nu\left((1+x)r^2-x\right).
\end{align*}
Since $r^2=-2\log((1+x)/2)$,
\[
        \E[A_*W]
        =-\nu\left(x+2(1+x)\log\frac{1+x}{2}\right)
        =-\nu\Delta(x).
\]
This proves the lemma.
\end{proof}


\begin{corollary}[Diagonal bound]
For the general-dimensional ternary $h$ above,
\[
        \norm{P_3h}^2\ge \frac{\nu^2}{6}\Delta_+(x)^2.
\]
Consequently
\[
        2\norm{P_1h}^2-\norm{P_3h}^2
        \le
        \nu^2\left(2x^2-\frac16\Delta_+(x)^2\right).
\]
\end{corollary}

\subsection{The disagreement term}
\label{sec:disagreement}

Now we have to bound the disagreement term
\(
        \norm{P_3k}^2.
\)
For each fixed vector \(t\), restrict \(k\) to the one-dimensional line obtained by varying \(S\) while holding \(T=t\). We call this restriction the \textbf{fiber} of \(k\) over \(t\), and write
\[
    u_t(s):=k(s,t).
\]
Since \(k\in\{-1,0,1\}\), every fiber \(u_t\) is a one-dimensional ternary function. For $j=0,1,2,3$, define
\[
        K_j(t):=\E_S[k(S,t)\HermBasis_j(S)].
\]
Thus $K_j$ is the $j$th Hermite coefficient of the $S$-fiber, as a function of the transverse variable $t$. In two dimensions, the four degree-three basis elements were
\[
        \HermBasis_3(S),\qquad
        \HermBasis_2(S)\HermBasis_1(T),\qquad
        \HermBasis_1(S)\HermBasis_2(T),\qquad
        \HermBasis_3(T).
\]
In general dimensions, this becomes the orthogonal decomposition
\[
        \cH_3(S,T)
        =\bigoplus_{j=0}^3 \cH_j(S)\otimes \cH_{3-j}(T).
\]
This just says that a total degree-three Hermite can spend $j$ degrees in $S$ and $3-j$ degrees in $T$.

Therefore
\[
        P_3k
        =\sum_{j=0}^3 \HermBasis_j(S)\,\Pi^T_{3-j}K_j(T),
\]
where $\Pi^T_{3-j}$ is the orthogonal projection in the $T$ variable onto transverse Hermite degree $3-j$. Hence
\begin{equation}
        \norm{P_3k}^2
        =\sum_{j=0}^3 \norm{\Pi^T_{3-j}K_j}^2.
        \label{eq:p3k-decomposition}
\end{equation}
Since the $\Pi^T_{3-j}$ are orthogonal projections onto closed subspaces, they decrease norm, so we get the simpler upper bound
\[
        \norm{P_3k}^2
        \le \sum_{j=0}^3\norm{K_j}^2
        =\E_T\sum_{j=0}^3\left(\E_S[k(S,T)\HermBasis_j(S)]\right)^2.
\]
For a one dimensional ternary function $u$, define
\[
        V(u):=\sum_{j=0}^3\left(\E[u(Z)\HermBasis_j(Z)]\right)^2.
\]
Then the previous inequality becomes
\begin{equation}
        \boxed{\norm{P_3k}^2\le \E_T V(u_T).}
        \label{eq:p3k-fiber-bound}
\end{equation}

For a ternary function $u:\R\to\{-1,0,1\}$, define its weighted support budget
\[
        \beta(u):=\E[|Z|u(Z)^2].
\]
Since $u^2=\one_{u\neq0}$,
\[
        \beta(u)=\E[|Z|\one_{u(Z)\neq0}].
\]
So $\beta(u)$ is the $|Z|$-weighted Gaussian size of the support of the fiber.

\begin{theorem}[Fiber inequality]
For every ternary $u:\R\to\{-1,0,1\}$,
\begin{equation}
        V(u)\le 3\nu\beta(u)-\beta(u)^2.
        \label{eq:fiber-inequality}
\end{equation}
\end{theorem}

\paragraph{Proof sketch}. This is proved in Appendix~\ref{app:fiber-certificate}, which reduces ~\eqref{eq:fiber-inequality} to a finite computational check. Here, we briefly sketch out the proof. The inequality is genuinely finite-dimensional: $V(u)$ depends only on the four moments $\E[u\HermBasis_j]$, $0\le j\le3$. By the dual formula for the Euclidean norm,
\[
        \sqrt{V(u)}=\sup_{\norm{a}_2=1}\E[u(Z)p_a(Z)],
        \qquad
        p_a(Z)=\sum_{j=0}^3a_j\HermBasis_j(Z).
\]
For fixed $p_a$ and a fixed support price $\tau\ge0$, pointwise maximization over the three values of $u$ gives
\[
        \max_{v\in\{-1,0,1\}}\bigl(v\,p_a(z)-\tau|z|\,v^2\bigr)
        =\bigl(|p_a(z)|-\tau|z|\bigr)_+,
\]
with maximizer the threshold rule
\[
        u(z)=\sgn(p_a(z))\one_{\{|p_a(z)|\ge \tau |z|\}};
\]
threshold functions are thus the exact optimizers of the priced problem, not merely plausible candidates. The infinite-dimensional search over $u$ therefore collapses to the four-parameter family $(a,\tau)$, whose active sets are cut out by one-dimensional polynomial inequalities; the content of~\eqref{eq:fiber-inequality} is the inequality $3\nu\beta-m^2\ge\beta^2$ along this family, where $m=\int_E|p_a|\,d\gamma$ and $\beta=\int_E|z|\,d\gamma$ over the active set $E$. The constant fiber $u\equiv1$ (that is, $p_a=\HermBasis_0$, $\tau=0$) realizes $\rho:=(3\nu\beta-m^2)/\beta^2=3-\pi/2\approx1.4292$, and numerical scans over the family find no smaller value. This suggests that~\eqref{eq:fiber-inequality} --- which asserts only $\rho\ge1$ --- holds with about $0.43$ of true headroom, though we do not prove this sharper assertion.

Appendix~\ref{app:fiber-certificate} turns this into a proof: it reduces~\eqref{eq:fiber-inequality} to two certified regimes (a direction-free high-budget bound and a Fenchel-dual low-budget bound), each consisting of finitely many one-dimensional inequalities that are verified in outward-rounded interval arithmetic. 

\paragraph{Reproducibility.}
  The scripts, certificates, and Arb outputs for this 
  are available at \url{https://github.com/trishullab/grothendieck-bounds}. 

\vspace{\baselineskip}

Apply~\eqref{eq:fiber-inequality} to each fiber $u_T$. From~\eqref{eq:p3k-fiber-bound},
\[
        \norm{P_3k}^2
        \le
        \E_T\left[3\nu\beta(u_T)-\beta(u_T)^2\right].
\]
We now prove the following bound. Note that this step is dimension-free.

\begin{theorem}[Fiber Residual Bound \#1]
We have the fiber bound
\begin{equation}
        \norm{P_3k}^2
        \le
        \nu^2\left(3(1-x)-(1-x)^2\right).
        \label{eq:fiber-residual-one}
\end{equation}
\end{theorem}

\begin{proof}
We do this by relating the average budget to $x$. Since
\[
        \nu x=\E[hS]
\]
and
\[
        hS\le |S|h^2,
\]
we get
\begin{align*}
        \nu x
        &\le \E[|S|h^2]\\
        &=\E[|S|(1-k^2)]\\
        &=\nu-\E[|S|k^2].
\end{align*}
Therefore
\[
        \E[|S|k^2]\le \nu(1-x).
\]
But
\[
        \E[|S|k^2]
        =\E_T\E_S[|S|k(S,T)^2]
        =\E_T\beta(u_T).
\]
Let
\[
        b:=\E_T\beta(u_T).
\]
Then
\[
        b\le \nu(1-x).
\]

The function
\[
        \phi(t)=3\nu t-t^2
\]
is concave, so Jensen gives
\[
        \E_T\phi(\beta(u_T))\le \phi(\E_T\beta(u_T))=3\nu b-b^2.
\]
On $[0,\nu]$, $\phi$ is increasing. Hence, using $b\le \nu(1-x)$, we obtain~\eqref{eq:fiber-residual-one}.
\end{proof}

This is our first fiber residual bound. But this bound is not best for small $x$. We also use a second cruder bound.

\begin{theorem}[Fiber Residual Bound \#2]
    \begin{equation}
        \norm{P_3k}^2
        \le
        2\PhiG(\sqrt{-2\log x})-1.
        \label{eq:fiber-residual-two}
\end{equation}
\end{theorem}

\begin{proof}
Let
\[
        \rho(s):=\E_T[k(s,T)^2].
\]
Then $0\le\rho\le1$ and
\[
        \norm{k}^2_2=\E\rho(S).
\]
Also
\[
        \E[|S|\rho(S)]=\E[|S|k(S,T)^2]\le \nu(1-x).
\]
Since $P_3$ is an orthogonal projection,
\[
        \norm{P_3k}^2\le \norm{k}^2_2=\E\rho(S).
\]

Given a budget on $\E[|S|\rho(S)]$, the way to maximize $\E\rho(S)$ is to spend the budget where $|S|$ is smallest. Thus the extremal support is centered around zero.

Let
\[
        r_x:=\sqrt{-2\log x}
\]
for $x\in(0,1]$, with the convention $r_0=\infty$. Then
\[
        \E[|S|\one_{|S|\le r_x}]=\nu(1-e^{-r_x^2/2})=\nu(1-x).
\]
Therefore the bathtub principle gives
\[
        \norm{P_3k}^2
        \le
        \mathbb P(|S|\le r_x)
        =2\PhiG(\sqrt{-2\log x})-1,
\]
which proves~\eqref{eq:fiber-residual-two}.
\end{proof}

Combining the two fiber residual bounds, define
\begin{equation}
        R(x):=
        \min\left\{
        3(1-x)-(1-x)^2,
        \frac{1}{\nu^2}\left(2\PhiG(\sqrt{-2\log x})-1\right)
        \right\}.
        \label{eq:residual-function}
\end{equation}
Then
\[
        \norm{P_3k}^2\le \nu^2R(x).
\]

\subsection{The final inequality}
\label{sec:putting-together}

From the diagonal bound and the residual bound,
\[
        D_1(h,k)
        \le
        \nu^2\left(2x^2-\frac16\Delta_+(x)^2+R(x)\right).
\]
So it remains to prove the scalar inequality
\begin{equation}
        E_1(x):=2x^2-\frac16\Delta_+(x)^2+R(x)
        \le \frac{11}{6}
        \qquad(0\le x\le1).
        \label{eq:scalar-reduction}
\end{equation}

Let
\[
        r_-:=\frac{1-1/\sqrt3}{2},
        \qquad
        r_+:=\frac{1+1/\sqrt3}{2}.
\]
These are the roots of
\[
        x^2-x+\frac16=0.
\]

Now we do some simple casework based on the range in which $x$ lies. Everything is simple algebra.

\subsection*{Low range: $0\le x\le r_-$}

Use the crude bound $R(x)\le 1/\nu^2=\pi/2$ and drop the negative term:
\[
        E_1(x)\le 2x^2+\frac\pi2
        \le 2r_-^2+\frac\pi2.
\]
Now
\[
        2r_-^2+\frac\pi2
        =\frac23-\frac1{\sqrt3}+\frac\pi2
        <\frac{11}{6}.
\]

\subsection*{Middle range: $r_-\le x\le r_+$}

Use the fiber branch of $R$ and again drop the negative term:
\begin{align*}
        E_1(x)
        &\le 2x^2+3(1-x)-(1-x)^2\\
        &=x^2-x+2.
\end{align*}
Since $x\in[r_-,r_+]$,
\[
        x^2-x+\frac16\le0.
\]
Therefore
\[
        x^2-x+2\le \frac{11}{6}.
\]

\subsection*{High range: $r_+\le x\le1$}

Use the fiber branch again. We need the cubic correction.
It suffices to show
\begin{equation}
        \Delta(x)^2\ge 6x^2-6x+1.
        \label{eq:high-range-correction}
\end{equation}
Indeed, then
\begin{align*}
        E_1(x)
        &\le 2x^2+3(1-x)-(1-x)^2-\frac16(6x^2-6x+1)\\
        &=\frac{11}{6}.
\end{align*}

To prove~\eqref{eq:high-range-correction}, first note that
\[
        \Delta(x)-(3x-2)
        =2\left((1-x)+(1+x)\log\frac{1+x}{2}\right).
\]
Writing $t=(1+x)/2$, the bracket becomes
\[
        2\bigl((1-t)+t\log t\bigr)\ge0
\]
for $0<t\le1$. Hence
\[
        \Delta(x)\ge 3x-2.
\]
For $x\ge r_+$, both sides needed below are nonnegative, and
\[
        (3x-2)^2-(6x^2-6x+1)=3(1-x)^2\ge0.
\]
Thus
\[
        \Delta(x)^2\ge (3x-2)^2\ge 6x^2-6x+1.
\]
This proves the high range.

Combining the three ranges proves~\eqref{eq:scalar-reduction}.

We have thus shown
\[
        D_1(h,k)\le \frac{11}{6}\nu^2.
\]
As explained earlier, this implies the affine coefficient inequality
\[
        b_3\ge 2b_1-\frac{11}{6}.
\]
Since this is affine, it is preserved under mixtures.

\section{The Affine Coefficient Inequality Bounds the Inverse Majorant}
\label{sec:strip_majorant}

\begin{definition}[Admissible majorant value]\label{def:admissible-majorant}
Let $H$ be odd and analytic near the origin, with $H'(0)>0$. Let $H^{-1}$ denote the local inverse branch satisfying $H^{-1}(0)=0$, and write
\[
        H^{-1}(z)=\sum_{j\ge0}a_{2j+1}z^{2j+1}.
\]
A number $\gamma>0$ is \emph{admissible} for $H$ if this series converges absolutely at $\gamma$ and
\[
        \sum_{j\ge0}|a_{2j+1}|\gamma^{2j+1}\le1.
\]
\end{definition}

We now turn the affine coefficient inequality proved in Section~\ref{sec:affine-strip} into a universal barrier for the inverse-majorant parameter. The coefficients $b_1$ and $b_3$ depend linearly on the correlation function, so the strip is preserved under mixtures and coefficientwise limits. Section~\ref{sec:majorant_optimal} will apply this barrier to the limiting mixed functions supplied by the Naor--Regev construction.

\begin{proposition}[The majorant barrier]\label{prop:majorant-barrier}
Let
\[
        H(t)=b_1t+b_3t^3+b_5t^5+\cdots
\]
be odd and analytic near the origin. Suppose that
\[
        0<b_1\le1,
        \qquad
        b_3\ge2b_1-\frac{11}{6}.
\]
Then every admissible value for $H$ satisfies
\[
        \gamma\le\frac{11}{12}.
\]
\end{proposition}

\begin{proof}
Write
\[
        H^{-1}(z)=a_1z+a_3z^3+\cdots.
\]
Comparing the coefficients of $z$ and $z^3$ in $H(H^{-1}(z))=z$ gives
\[
        a_1=\frac1{b_1},
        \qquad
        a_3=-\frac{b_3}{b_1^4}.
\]

If $b_1\le11/12$, admissibility gives
\[
        1\ge |a_1|\gamma=\frac{\gamma}{b_1},
\]
so $\gamma\le11/12$.

Suppose now that $b_1>11/12$. The strip gives
\[
        b_3\ge2\left(b_1-\frac{11}{12}\right)>0.
\]
Let
\[
        M_H(s):=\sum_{j\ge0}|a_{2j+1}|s^{2j+1},
\]
allowing the value $+\infty$. Using only the first two terms,
\begin{align*}
        M_H\left(\frac{11}{12}\right)
        &\ge
        \frac{11}{12b_1}
        +\frac{2(b_1-\frac{11}{12})}{b_1^4}
          \left(\frac{11}{12}\right)^3\\
        &=1+
        \frac{(b_1-\frac{11}{12})
        \left(2(\frac{11}{12})^3-b_1^3\right)}{b_1^4}
        >1.
\end{align*}
The last inequality uses $b_1\le1$ and $2(11/12)^3>1$. If some $\gamma\ge11/12$ were admissible, monotonicity of $M_H$ would give
\[
        M_H\left(\frac{11}{12}\right)\le M_H(\gamma)\le1,
\]
a contradiction.
\end{proof}

\section{The Inverse Majorant Bound is Optimal}
\label{sec:majorant_optimal}

The key global input is obtained by reworking the construction in the proof of Naor and Regev's optimality theorem. We first record the notion of mixing needed to state it.

\begin{definition}[Mixed rounding scheme]\label{def:mixed-rounding}
Fix $d\in\mathbb N$. A $d$-dimensional \emph{mixed rounding scheme} is a probability distribution $\mu$ on pairs of measurable sign functions
\[
        f,g:\R^d\to\{\pm1\}.
\]
One pair $(f,g)$ is sampled from $\mu$ and then used for every vector in the rounding procedure. Its averaged correlation function is
\[
        H_\mu(t):=\E_{(f,g)\sim\mu}H_{f,g}(t).
\]
A \emph{limiting mixed correlation function} is a coefficientwise limit of averaged correlation functions of mixed rounding schemes, possibly in increasing dimensions.
\end{definition}

The admissible values from Definition~\ref{def:admissible-majorant} are the parameters appearing in the inverse-majorant bound~\eqref{eq:inverse-majorant-bound}. For an actual mixed scheme, the standard tensor preprocessing argument turns an admissible value $\gamma$ into the upper bound $K_G\le\pi/(2\gamma)$. The next lemma shows that, after allowing mixtures and increasing dimension, these bounds can approach $K_G$ itself.

\begin{lemma}[Optimality of the inverse-majorant bound]\label{lem:inverse-majorant-optimal}
There exist odd analytic limiting mixed correlation functions $H_k$ and admissible values $\gamma_k$ such that
\[
        \gamma_k\le\frac{\pi}{2K_G}
        \qquad\text{and}\qquad
        \gamma_k\longrightarrow\frac{\pi}{2K_G}.
\]
Equivalently, the resulting upper bounds satisfy
\[
        \frac{\pi}{2\gamma_k}\longrightarrow K_G.
\]
Thus, after allowing mixtures and increasing dimension, the inverse-majorant bound is asymptotically optimal.
\end{lemma}

\begin{proof}
We use the construction in the proof of Theorem~1.1 of \cite{naor2014krivine}. We do not use their theorem statement directly.

For each $k$, their equation~(4) gives a probability measure $\nu_k$ on pairs of sign assignments on $S^{k-1}$ whose averaged kernel is
\[
        \frac{1}{K_G}\ip{x}{y}.
\]
After radial extension, their equations~(5) and~(6) give an odd analytic normalized correlation function
\[
        H_k(t)=\frac{\pi}{2K_G}f_k(t),
\]
where $f_k$ is the function defined in their equation~(6). By partitioning the sphere into finitely many small Borel sets and extending the sampled signs as step functions, one obtains ordinary mixed rounding schemes whose averaged kernels converge uniformly to the kernel above. Hence $H_k$ is a limiting mixed correlation function in the sense of Definition~\ref{def:mixed-rounding}.

Naor and Regev write
\[
        f_k^{-1}(z)=\sum_{n\ge0}b_n(k)z^{2n+1}.
\]
Immediately before their Lemma~2.1, they choose a number $c_k\le1$ such that
\[
        \sum_{n\ge0}|b_n(k)|c_k^{2n+1}\le1,
\]
and Lemma~2.1 proves that $c_k\to1$. Define
\[
        \gamma_k:=\frac{\pi c_k}{2K_G}.
\]
Since
\[
        H_k^{-1}(z)=f_k^{-1}\left(\frac{2K_G}{\pi}z\right),
\]
the absolute inverse majorant of $H_k$ at $\gamma_k$ is exactly
\[
        \sum_{n\ge0}|b_n(k)|c_k^{2n+1}\le1.
\]
Thus $\gamma_k$ is admissible. Since $c_k\le1$ and $c_k\to1$,
\[
        \gamma_k\le\frac{\pi}{2K_G}
        \qquad\text{and}\qquad
        \gamma_k\longrightarrow\frac{\pi}{2K_G}.
\]
The equivalent statement about the resulting upper bounds follows by taking reciprocals.
\end{proof}

\begin{proof}[Proof of Theorem~\ref{thm:lower-bound}]
Let $H_k$ and $\gamma_k$ be supplied by Lemma~\ref{lem:inverse-majorant-optimal}. Sections~\ref{sec:preliminaries} and~\ref{sec:affine-strip} establish the affine coefficient inequality for every pair of odd sign functions in every dimension. It therefore holds for every mixed scheme, because $b_1$ and $b_3$ average linearly. Passing to the coefficientwise limit shows that
\[
        b_3(H_k)\ge2b_1(H_k)-\frac{11}{6}.
\]
This is the point at which preservation under mixing is used.

We also have $0<b_1(H_k)\le1$. Positivity follows because $\gamma_k$ is admissible. For the upper bound, every sign function $f$ satisfies
\[
        \norm{P_1f}
        =\sup_{\norm{u}=1}\left|\E[f(X)\ip{X}{u}]\right|
        \le\E|S|=\nu,
\]
so $|b_1(f,g)|\le(\pi/2)\nu^2=1$. This bound is also preserved under mixtures and coefficientwise limits.

Proposition~\ref{prop:majorant-barrier} now gives
\[
        \gamma_k\le\frac{11}{12}.
\]
Letting $k\to\infty$ and using Lemma~\ref{lem:inverse-majorant-optimal},
\[
        \frac{\pi}{2K_G}\le\frac{11}{12}.
\]
Rearranging proves the result.
\end{proof}

%% file: paper/statement_ai_use.tex
  \paragraph{Statement on AI use.}
  The research reported in this paper made extensive use of large language      
  models: GPT-5.5-Pro and GPT-5.6-Sol (OpenAI) served as the reasoning          
  models, and Claude Opus and Claude Fable~5 (Anthropic) as the coding          
  agents, operating within the AI research system described in the companion    paper~\cite{lisaha2026longhorizon}.               
  AI assistance was also used in preparing parts of this manuscript.            
  Every result stated as a theorem has been independently verified by the       
  authors, and the authors take full responsibility for the contents of this    
  paper. 

%% file: paper/appendix.tex
\newpage
\appendix

\section{Krivine Preprocessing and Mixed Rounding}
\label[appendix]{app:krivine-preprocessing}

This appendix records the preprocessing and projection arguments behind
Krivine rounding schemes.
The argument is the standard rounding construction
~\cite{krivine1977constante, braverman2011grothendieckconstantstrictlysmaller}.

For a real matrix \(A=(a_{ij})\in\mathbb R^{m\times n}\), write
\[
    \operatorname{SDP}(A)
    :=
    \sup_{\|x_i\|=\|y_j\|=1}
    \left|
        \sum_{i=1}^m\sum_{j=1}^n a_{ij}\langle x_i,y_j\rangle
    \right|
\]
and
\[
    \operatorname{OPT}(A)
    :=
    \max_{\varepsilon_i,\delta_j\in\{\pm 1\}}
    \left|
        \sum_{i=1}^m\sum_{j=1}^n a_{ij}\varepsilon_i\delta_j
    \right|.
\]
Thus proving \(\operatorname{SDP}(A)\le K\operatorname{OPT}(A)\) for all
\(A\) proves \(K_G\le K\).

\begin{theorem}[Krivine preprocessing and projection]
\label{thm:appendix-pure-krivine-rounding}
Let \(f,g:\mathbb R^k\to\{-1,1\}\) be odd measurable functions, and let
\(H=H_{f,g}\) be the arcsine-normalized correlation function
\[
    H(t)
    =
    \frac{\pi}{2}
    \mathbb E\left[
        f\left(G_1\right)
        g\left(tG_1+\sqrt{1-t^2}G_2\right)
    \right],
    \qquad -1<t<1,
\]
where \(G_1,G_2\) are independent standard Gaussian vectors in
\(\mathbb R^k\). Suppose that \(H\) has a local inverse at the origin,
\[
    H^{-1}(\zeta)=\sum_{r\ge 1} a_r\zeta^r,
\]
which is holomorphic on a neighborhood of \(\gamma\overline{\mathbb D}\),
and suppose that
\[
    \sum_{r\ge 1}|a_r|\gamma^r\le 1.
\]
Then
\[
    K_G\le \frac{\pi}{2\gamma}.
\]
\end{theorem}

\begin{proof}
It is enough to prove the conclusion with any \(0<\gamma'<\gamma\) in
place of \(\gamma\), and then let \(\gamma'\uparrow\gamma\). We therefore
fix such a \(\gamma'\). Put
\[
    M_{\gamma'}:=\sum_{r\ge 1}|a_r|(\gamma')^r < 1.
\]
Let \(A=(a_{ij})\in\mathbb R^{m\times n}\). Choose unit vectors
\(x_1,\ldots,x_m,y_1,\ldots,y_n\) such that
\[
    \sum_{i,j} a_{ij}\langle x_i,y_j\rangle
\]
is arbitrarily close to \(\operatorname{SDP}(A)\); if necessary, replace all
\(y_j\) by \(-y_j\) so that this quantity is nonnegative. We suppress this
arbitrarily small error, since it disappears at the end.

Let
\[
    \mathcal K
    :=
    \left(\bigoplus_{r\ge 1} \mathcal H^{\otimes r}\right)
    \oplus \mathbb R e_L\oplus \mathbb R e_R,
\]
where \(\mathcal H\) is the Hilbert space containing the original vectors,
and where \(e_L,e_R\) are unit vectors orthogonal to each other and to all
tensor summands. Define
\[
    I(x)
    :=
    \bigoplus_{r\ge 1}
        \sqrt{|a_r|}\,(\gamma')^{r/2}x^{\otimes r}
    \oplus \sqrt{1-M_{\gamma'}}\,e_L
    \oplus 0
\]
and
\[
    J(y)
    :=
    \bigoplus_{r\ge 1}
        \operatorname{sgn}(a_r)\sqrt{|a_r|}\,(\gamma')^{r/2}y^{\otimes r}
    \oplus 0
    \oplus \sqrt{1-M_{\gamma'}}\,e_R,
\]
with the convention that the summand corresponding to \(a_r=0\) is zero.
Then \(I(x)\) and \(J(y)\) are unit vectors whenever \(x,y\) are unit
vectors. Moreover,
\[
    \langle I(x),J(y)\rangle
    =
    \sum_{r\ge 1} a_r(\gamma')^r\langle x,y\rangle^r
    =
    H^{-1}\bigl(\gamma'\langle x,y\rangle\bigr).
\]
The first equality uses
\[
    \langle x^{\otimes r},y^{\otimes r}\rangle
    =
    \langle x,y\rangle^r,
\]
and the second equality uses the Taylor expansion of the inverse on
\(\gamma'\mathbb D\).

Set
\[
    u_i:=I(x_i),
    \qquad
    v_j:=J(y_j).
\]
The finitely many vectors \(u_i,v_j\) span a finite-dimensional subspace of
\(\mathcal K\). Identify this subspace with \(\mathbb R^N\). Let
\(\Gamma:\mathbb R^N\to\mathbb R^k\) be a random Gaussian linear map, i.e.
a \(k\times N\) matrix whose entries are independent standard Gaussian
random variables. Define random signs
\[
    \varepsilon_i
    :=
    f\left(\Gamma u_i\right),
    \qquad
    \delta_j
    :=
    g\left(\Gamma v_j\right).
\]
For each pair \((i,j)\), the vectors \(\Gamma u_i\) and \(\Gamma v_j\) are
standard Gaussian vectors in \(\mathbb R^k\) with coordinatewise correlation
\(\langle u_i,v_j\rangle\). Hence, by the definition of \(H\),
\[
    \mathbb E[\varepsilon_i\delta_j]
    =
    \frac{2}{\pi}
    H(\langle u_i,v_j\rangle).
\]
Using the preprocessing identity,
\[
    H(\langle u_i,v_j\rangle)
    =
    H\left(H^{-1}\bigl(\gamma'\langle x_i,y_j\rangle\bigr)\right)
    =
    \gamma'\langle x_i,y_j\rangle.
\]
Therefore
\[
\begin{aligned}
    \mathbb E\left[\sum_{i,j} a_{ij}\varepsilon_i\delta_j\right]
    &=
    \frac{2}{\pi}\sum_{i,j}a_{ij}
        H(\langle u_i,v_j\rangle) \\
    &=
    \frac{2\gamma'}{\pi}
    \sum_{i,j}a_{ij}\langle x_i,y_j\rangle.
\end{aligned}
\]
Since the maximum over deterministic signs is at least the expectation of
the randomized signs, we get
\[
    \operatorname{OPT}(A)
    \ge
    \frac{2\gamma'}{\pi}\operatorname{SDP}(A).
\]
Thus
\[
    \operatorname{SDP}(A)
    \le
    \frac{\pi}{2\gamma'}\operatorname{OPT}(A).
\]
Letting \(\gamma'\uparrow\gamma\) gives
\[
    \operatorname{SDP}(A)
    \le
    \frac{\pi}{2\gamma}\operatorname{OPT}(A).
\]
Since \(A\) was arbitrary, \(K_G\le \pi/(2\gamma)\).
\end{proof}

\begin{corollary}[Classical Krivine hyperplane bound]
\label{cor:appendix-classical-krivine}
Let \(h:\mathbb R\to\{-1,1\}\) be \(h(x)=\operatorname{sgn}(x)\). Then the
hyperplane scheme \((h,h)\) gives
\[
    K_G\le \frac{\pi}{2\log(1+\sqrt 2)}.
\]
\end{corollary}

\begin{proof}
For the hyperplane scheme, Grothendieck's identity gives
\[
    \mathbb E[\operatorname{sgn}(X)\operatorname{sgn}(Y)]
    =
    \frac{2}{\pi}\arcsin(t)
\]
whenever \(X,Y\) are standard Gaussians with correlation \(t\). In our
normalization,
\[
    H_{h,h}(t)=\arcsin(t).
\]
Thus
\[
    H_{h,h}^{-1}(\zeta)=\sin \zeta
    =
    \sum_{q\ge 0}\frac{(-1)^q}{(2q+1)!}\zeta^{2q+1}.
\]
Let
\[
    \rho_*:=\operatorname{arsinh}(1)=\log(1+\sqrt 2).
\]
Then
\[
    \sum_{q\ge 0}
    \frac{\rho_*^{2q+1}}{(2q+1)!}
    =
    \sinh(\rho_*)
    =
    1.
\]
Theorem~\ref{thm:appendix-pure-krivine-rounding} therefore applies with
\(\gamma=\rho_*\), giving
\[
    K_G\le \frac{\pi}{2\rho_*}
    =
    \frac{\pi}{2\log(1+\sqrt 2)}.
\]
\end{proof}

\begin{theorem}[Mixed Krivine preprocessing and projection]
\label{thm:appendix-mixed-krivine-rounding}
Let \(L\in\mathbb N\). For each \(\ell\in\{1,\ldots,L\}\), let
\(f_\ell,g_\ell:\mathbb R^{k_\ell}\to\{-1,1\}\) be odd measurable functions,
and let \(H_\ell\) be the corresponding arcsine-normalized correlation
function:
\[
    H_\ell(t)
    =
    \frac{\pi}{2}
    \mathbb E\left[
        f_\ell\left(G_1\right)
        g_\ell\left(tG_1+\sqrt{1-t^2}G_2\right)
    \right].
\]
Let \(\lambda_1,\ldots,\lambda_L\ge 0\) satisfy
\(\sum_{\ell=1}^L\lambda_\ell=1\), and define the mixed function
\[
    H_\lambda(t):=\sum_{\ell=1}^L\lambda_\ell H_\ell(t).
\]
Suppose that \(H_\lambda\) has a local inverse at the origin,
\[
    H_\lambda^{-1}(\zeta)=\sum_{r\ge 1} a_r\zeta^r,
\]
which is holomorphic on a neighborhood of \(\gamma\overline{\mathbb D}\),
and suppose that
\[
    M_\lambda(\gamma)
    :=
    \sum_{r\ge 1}|a_r|\gamma^r
    \le 1.
\]
Then
\[
    K_G\le \frac{\pi}{2\gamma}.
\]
\end{theorem}

\begin{proof}
The preprocessing is identical to the proof of
Theorem~\ref{thm:appendix-pure-krivine-rounding}, but using the inverse
coefficients of \(H_\lambda^{-1}\). As before, first fix
\(0<\gamma'<\gamma\), and set
\[
    M_{\lambda,\gamma'}
    :=
    \sum_{r\ge 1}|a_r|(\gamma')^r
    <1.
\]
For unit vectors \(x,y\), define
\[
    I(x)
    :=
    \bigoplus_{r\ge 1}
        \sqrt{|a_r|}\,(\gamma')^{r/2}x^{\otimes r}
    \oplus \sqrt{1-M_{\lambda,\gamma'}}\,e_L
    \oplus 0
\]
and
\[
    J(y)
    :=
    \bigoplus_{r\ge 1}
        \operatorname{sgn}(a_r)\sqrt{|a_r|}\,(\gamma')^{r/2}y^{\otimes r}
    \oplus 0
    \oplus \sqrt{1-M_{\lambda,\gamma'}}\,e_R.
\]
Then \(I(x)\) and \(J(y)\) are unit vectors and
\[
    \langle I(x),J(y)\rangle
    =
    H_\lambda^{-1}\bigl(\gamma'\langle x,y\rangle\bigr).
\]

Now let \(A=(a_{ij})\) be arbitrary, and choose unit vectors
\(x_i,y_j\) that nearly attain \(\operatorname{SDP}(A)\), oriented so that
\(\sum_{i,j}a_{ij}\langle x_i,y_j\rangle\ge 0\). Put
\[
    u_i:=I(x_i),
    \qquad
    v_j:=J(y_j).
\]
Identify the span of the finitely many \(u_i,v_j\) with \(\mathbb R^N\).
For each \(\ell\), independently sample a Gaussian linear map
\[
    \Gamma_\ell:\mathbb R^N\to\mathbb R^{k_\ell}.
\]
Also sample a single global index \(\Lambda\in\{1,\ldots,L\}\), independent
of the Gaussian maps, with
\[
    \mathbb P[\Lambda=\ell]=\lambda_\ell.
\]
The word ``global'' is important: the same index \(\Lambda\) is used for
all \(i,j\). Define
\[
    \varepsilon_i
    :=
    f_\Lambda\left(\Gamma_\Lambda u_i\right),
    \qquad
    \delta_j
    :=
    g_\Lambda\left(\Gamma_\Lambda v_j\right).
\]
Conditioning on \(\Lambda=\ell\), the same rotation-invariance calculation
as in the pure case gives
\[
    \mathbb E[\varepsilon_i\delta_j\mid \Lambda=\ell]
    =
    \frac{2}{\pi}H_\ell(\langle u_i,v_j\rangle).
\]
Averaging over \(\Lambda\),
\[
    \mathbb E[\varepsilon_i\delta_j]
    =
    \frac{2}{\pi}
    \sum_{\ell=1}^L\lambda_\ell H_\ell(\langle u_i,v_j\rangle)
    =
    \frac{2}{\pi}H_\lambda(\langle u_i,v_j\rangle).
\]
Since
\[
    \langle u_i,v_j\rangle
    =
    H_\lambda^{-1}\bigl(\gamma'\langle x_i,y_j\rangle\bigr),
\]
we obtain
\[
    \mathbb E[\varepsilon_i\delta_j]
    =
    \frac{2\gamma'}{\pi}\langle x_i,y_j\rangle.
\]
Therefore
\[
\begin{aligned}
    \operatorname{OPT}(A)
    &\ge
    \mathbb E\left[
        \sum_{i,j}a_{ij}\varepsilon_i\delta_j
    \right] \\
    &=
    \frac{2\gamma'}{\pi}
    \sum_{i,j}a_{ij}\langle x_i,y_j\rangle.
\end{aligned}
\]
Taking the supremum over the original vectors gives
\[
    \operatorname{OPT}(A)
    \ge
    \frac{2\gamma'}{\pi}\operatorname{SDP}(A),
\]
and hence
\[
    \operatorname{SDP}(A)
    \le
    \frac{\pi}{2\gamma'}\operatorname{OPT}(A).
\]
Letting \(\gamma'\uparrow\gamma\) proves
\[
    \operatorname{SDP}(A)
    \le
    \frac{\pi}{2\gamma}\operatorname{OPT}(A).
\]
Since \(A\) was arbitrary, \(K_G\le \pi/(2\gamma)\).
\end{proof}

\begin{remark}
Theorem~\ref{thm:appendix-pure-krivine-rounding} is the special case
\(L=1\) of Theorem~\ref{thm:appendix-mixed-krivine-rounding}. We separated
the two statements because the pure case is the usual Krivine scheme, while
the mixed case emphasizes the randomized choice of a common scheme
\(\Lambda\) before projection.
\end{remark}

\section{General Coefficient Formulas}
\label{app:coefficient-formulas}

This appendix records two standard calculations used in
Section~\ref{sec:coefficients}. We first prove the closed form for the
one-dimensional threshold coefficients. We then give the coefficient
formula for a more general limiting scheme built from several Hermite
correlation coordinates. The latter is not needed for the main theorem,
but it explains how the calculation extends to other odd degrees.

\begin{lemma}[One-dimensional threshold coefficients]
\label{lem:qk}
Let \(Z\) be a standard Gaussian and define
\(
q_k(u)=\E[\sgn(Z+u)\He_k(Z)]
\).
Then
\[
    q_0(u)
    =
    \erf\left(\frac{u}{\sqrt{2}}\right),
\]
and, for \(k\geq 1\),
\[
    q_k(u)
    =
    \frac{2\varphi(u)}{\sqrt{k}}\He_{k-1}(-u),
\]
where \(\varphi(u)=(2\pi)^{-1/2}e^{-u^2/2}\).
\end{lemma}

\begin{proof}
The formula for \(q_0\) follows directly from the Gaussian distribution
function. For \(k\geq 1\), orthogonality to constants gives
\[
    q_k(u)
    =
    2\int_{-u}^{\infty}\He_k(z)\varphi(z)\,dz.
\]
The orthonormal probabilists' Hermite polynomials satisfy
\[
    \frac{d}{dz}
    \bigl(\varphi(z)\He_{k-1}(z)\bigr)
    =
    -\sqrt{k}\,\varphi(z)\He_k(z).
\]
Integrating from \(-u\) to \(\infty\) proves the claim.
\end{proof}

We next record the general formula. Let \(D\) be a finite set of odd
integers. Fix weights \((s_d)_{d\in D}\), set
\[
    S
    =
    \left(\sum_{d\in D}s_d^2\right)^{1/2},
    \qquad
    \alpha_d
    =
    \frac{s_d}{S},
\]
and let \(\sigma_d\in\{-1,1\}\). Consider the limiting Krivine scheme
with partitions
\[
    f_D(z,x,(r_d)_{d\in D})
    =
    \sgn\!\left(
        z+\eta\He_3(x)+\sum_{d\in D}s_dr_d
    \right),
\]
\[
    g_D(z,x,(r_d)_{d\in D})
    =
    \sgn\!\left(
        z-\eta\He_3(x)+\sum_{d\in D}\sigma_ds_dr_d
    \right),
\]
whose coordinate correlation maps are \(t,t\), and \(t^d\) for
\(d\in D\).

For \(a,b,k\geq 0\), define
\begin{equation}
    C_{a,b,k}
    =
    \E\!\left[
        \He_a(Z)\He_b(X)
        q_k\!\left(
            \frac{Z+\eta\He_3(X)}{S}
        \right)
    \right],
    \label{eq:general-Cabk}
\end{equation}
where \(Z,X\) are independent standard Gaussians. For a multi-index
\(\nu=(\nu_d)_{d\in D}\), write
\(
|\nu|=\sum_d\nu_d
\),
\(
\nu!=\prod_d\nu_d!
\), and
\(
\alpha^\nu=\prod_d\alpha_d^{\nu_d}
\).

\begin{proposition}[General degree formula]
\label{prop:general-degree-formula}
Let
\(
H_D(t)=\sum_{m\geq 1}b_m t^m
\)
be the correlation function of the scheme above. Then
\begin{equation}
    b_m
    =
    \frac{\pi}{2}
    \sum_{\substack{a,b\geq 0,\ \nu\in\mathbb{N}_0^D\\
          a+b+\sum_{d\in D}d\nu_d=m}}
    (-1)^b
    \left(\prod_{d\in D}\sigma_d^{\nu_d}\right)
    \frac{k!}{\nu!}
    \left(\prod_{d\in D}\alpha_d^{2\nu_d}\right)
    C_{a,b,k}^2,
    \qquad
    k=|\nu|.
    \label{eq:general-bm-appendix}
\end{equation}
Only odd \(m\) occur.
\end{proposition}

\begin{proof}
Set
\[
    T(Z,X)=Z+\eta\He_3(X),
    \qquad
    R_\alpha=\sum_{d\in D}\alpha_dR_d.
\]
Then \(f_D=\sgn(T+SR_\alpha)\). Conditioning on \(Z,X\), the coefficient
of \(\He_k(R_\alpha)\) is
\(
q_k(T(Z,X)/S)
\). Taking the inner product with
\(\He_a(Z)\He_b(X)\) gives \(C_{a,b,k}\).

The multivariate Hermite addition formula is
\[
    \He_k\!\left(\sum_{d\in D}\alpha_dR_d\right)
    =
    \sum_{\substack{\nu\in\mathbb{N}_0^D\\|\nu|=k}}
    \sqrt{\frac{k!}{\nu!}}\,
    \alpha^\nu
    \prod_{d\in D}\He_{\nu_d}(R_d).
\]
Hence the coefficient of
\(
\He_a(Z)\He_b(X)\prod_d\He_{\nu_d}(R_d)
\)
in \(f_D\) is
\[
    C_{a,b,k}
    \sqrt{\frac{k!}{\nu!}}\,
    \alpha^\nu.
\]
The corresponding coefficient of \(g_D\) differs by the factor
\(
(-1)^b\prod_d\sigma_d^{\nu_d}
\).

Under the limiting scheme, the matching Hermite term contributes
\[
    t^a t^b\prod_{d\in D}(t^d)^{\nu_d}
    =
    t^{a+b+\sum_d d\nu_d}.
\]
Multiplying the two Hermite coefficients and collecting the coefficient
of \(t^m\) gives \eqref{eq:general-bm-appendix}. Finally, both partitions
are odd, so their Hermite coefficients vanish unless
\(a+b+|\nu|\) is odd. Since every \(d\in D\) is odd, this has the same
parity as \(a+b+\sum_d d\nu_d\).
\end{proof}

\section{Certification Details}
\label{app:certification-details}

This appendix proves Proposition~\ref{prop:certified-numerical-input}.  The
first computation encloses the coefficient head through degree $251$.  The
second bounds the single boundary norm
$\|D^3H\|_{L^2(\mathbb T)}$ used in
Lemma~\ref{lem:third-order-tail}.

\subsection{The coefficient head}
\label{app:certification-coefficient-head}

Let $W,X$ be independent standard Gaussians and write
\[
  f(w,x)=\operatorname{sgn}\bigl(w+\vartheta\operatorname{He}_3(x)\bigr).
\]
For $a,b\ge0$, define
\[
  A_{a,b}
  :=\mathbb E\bigl[f(W,X)\operatorname{He}_a(W)
  \operatorname{He}_b(X)\bigr].
\]
As shown in Section~4, if
\[
  \rho(t)=\frac{t-s_3^2t^3+s_5^2t^5}{1+s_3^2+s_5^2},
\]
then
\begin{equation}
  H(t)=\frac{\pi}{2}\sum_{a,b\ge0}A_{a,b}^2\rho(t)^a(-t)^b,
  \label{eq:appendix-correlation-expansion}
\end{equation}
and hence
\begin{equation}
  b_m=\frac{\pi}{2}
  \sum_{\substack{a,b\ge0\\a+b\le m}}
  (-1)^bA_{a,b}^2[t^{m-b}]\rho(t)^a.
  \label{eq:appendix-coefficient-formula}
\end{equation}

The coefficients $A_{a,b}$ reduce to one-dimensional Gaussian integrals.  If
$Z$ is standard Gaussian and
\[
  q_a(u):=\mathbb E[\operatorname{sgn}(Z+u)\operatorname{He}_a(Z)],
\]
then
\[
  A_{a,b}
  =\mathbb E\bigl[\operatorname{He}_b(X)
  q_a(\vartheta\operatorname{He}_3(X))\bigr],
\]
where
\[
  q_0(u)=\operatorname{erf}\!\left(\frac{u}{\sqrt2}\right),
  \qquad
  q_a(u)=\frac{2\varphi(u)}{\sqrt a}
  \operatorname{He}_{a-1}(-u)
  \quad(a\ge1).
\]
Each integral is enclosed by outward-rounded adaptive quadrature on a finite
interval, together with a closed Gaussian-tail estimate.  Finite coefficient
extraction in \eqref{eq:appendix-coefficient-formula} then gives
\begin{equation}
  b_1\ge0.881573822049,
  \qquad
  \sum_{\substack{3\le m\le251\\m\ \mathrm{odd}}}|b_m|
  \le1.1328860\cdot10^{-5}.
  \label{eq:appendix-head-output}
\end{equation}

\subsection{A boundary representation of $D^3H$}
\label{app:certification-boundary-representation}

For real $|r|<1$, let
\[
  K_r(u,v)
  :=\frac{1}{2\pi\sqrt{1-r^2}}
  \exp\!\left(-\frac{u^2-2ruv+v^2}{2(1-r^2)}\right)
\]
be the centered bivariate Gaussian density with correlation $r$.  Let
$C_r(x,y)$ denote the conditional sign-correlation kernel
\[
  C_r(x,y)
  :=\frac{\pi}{2}\,
  \mathbb E\!\left[
  \operatorname{sgn}(W+\vartheta\operatorname{He}_3(x))
  \operatorname{sgn}(W'+\vartheta\operatorname{He}_3(y))
  \right],
\]
where $(W,W')$ is a standard Gaussian pair with correlation $r$.
The two derivative identities needed below follow directly from the Hermite
expansion of this kernel.

\begin{lemma}[Differentiating the threshold-correlation kernel]
\label{lem:appendix-kernel-differentiation}
For real $|r|<1$,
\begin{align}
  \partial_rC_r(x,y)
  &=2\pi K_r\bigl(\vartheta\operatorname{He}_3(x),
                   \vartheta\operatorname{He}_3(y)\bigr),
  \label{eq:appendix-kernel-r}\\
  \partial_x\partial_yC_r(x,y)
  &=2\pi\vartheta^2\operatorname{He}_3'(x)
  \operatorname{He}_3'(y)
  K_r\bigl(\vartheta\operatorname{He}_3(x),
           \vartheta\operatorname{He}_3(y)\bigr).
  \label{eq:appendix-kernel-xy}
\end{align}
\end{lemma}

\begin{proof}
Put
\[
  u=\vartheta\operatorname{He}_3(x),
  \qquad
  v=\vartheta\operatorname{He}_3(y).
\]
The Hermite covariance identity gives
\begin{equation}
  C_r(x,y)=\frac{\pi}{2}\sum_{a\ge0}r^a q_a(u)q_a(v),
  \label{eq:appendix-Cr-Hermite}
\end{equation}
where $q_a$ is the one-dimensional threshold coefficient defined above.  The
series and its $r$-derivative converge locally uniformly for $|r|<1$.  For
$a\ge1$,
\[
  q_a(u)=\frac{2\varphi(u)}{\sqrt a}\operatorname{He}_{a-1}(-u).
\]
Consequently,
\begin{align*}
  \partial_rC_r(x,y)
  &=\frac{\pi}{2}\sum_{a\ge1}a r^{a-1}q_a(u)q_a(v)\\
  &=2\pi\varphi(u)\varphi(v)
    \sum_{n\ge0}r^n\operatorname{He}_n(-u)
                         \operatorname{He}_n(-v).
\end{align*}
Mehler's identity identifies the last expression with $2\pi K_r(u,v)$,
proving \eqref{eq:appendix-kernel-r}.

Differentiating the defining integral for $q_a$ gives, for every $a\ge0$,
\[
  q_a'(u)=2\varphi(u)\operatorname{He}_a(-u).
\]
Differentiating \eqref{eq:appendix-Cr-Hermite} in $x$ and $y$ and applying
Mehler's identity once more yields
\begin{align*}
  \partial_x\partial_yC_r(x,y)
  &=\frac{\pi}{2}\vartheta^2\operatorname{He}_3'(x)
    \operatorname{He}_3'(y)
    \sum_{a\ge0}r^a q_a'(u)q_a'(v)\\
  &=2\pi\vartheta^2\operatorname{He}_3'(x)
    \operatorname{He}_3'(y)K_r(u,v),
\end{align*}
which proves \eqref{eq:appendix-kernel-xy}.
\end{proof}

Thus, after one derivative, the discontinuous sign functions are replaced by
explicit smooth Gaussian kernels.  The formulas above initially hold for
$|r|<1$; the explicit Gaussian expressions provide the analytic continuation
used on the complex curve $r=\rho(e^{i\theta})$ below.

If $\Psi$ has Hermite coefficients $\widehat\Psi_{i,j}$, define the diagonal
trace
\[
  T_\lambda[\Psi]
  :=\widehat\Psi_{0,0}+\sum_{n\ge1}\lambda^n\widehat\Psi_{n,n},
  \qquad |\lambda|\le1.
\]
For $|\lambda|<1$, this is the expectation of $\Psi(X,Y)$ when $(X,Y)$ is a
standard Gaussian pair with correlation $\lambda$.

\begin{lemma}[Differentiating the diagonal trace]
\label{lem:appendix-trace-differentiation}
If $\Psi(t)$ is a sufficiently regular family of kernels, then for $|t|<1$,
\begin{equation}
  D_tT_{-t}[\Psi(t)]
  =T_{-t}\!\left[D_t\Psi(t)-t\partial_x\partial_y\Psi(t)\right],
  \qquad D_t=t\partial_t.
  \label{eq:appendix-trace-differentiation}
\end{equation}
\end{lemma}

\begin{proof}
Write
\[
  \Psi(t;x,y)=\sum_{i,j\ge0}\widehat\Psi_{i,j}(t)
  \operatorname{He}_i(x)\operatorname{He}_j(y).
\]
Since $D_t(-t)^n=n(-t)^n$, termwise differentiation gives
\[
  D_tT_{-t}[\Psi(t)]
  =T_{-t}[D_t\Psi(t)]
   +\sum_{n\ge1}n(-t)^n\widehat\Psi_{n,n}(t).
\]
The identity
$\frac{d}{dx}\operatorname{He}_n(x)=\sqrt n\operatorname{He}_{n-1}(x)$
shows that
\[
  -tT_{-t}[\partial_x\partial_y\Psi(t)]
  =\sum_{n\ge1}n(-t)^n\widehat\Psi_{n,n}(t).
\]
Combining the last two displays proves
\eqref{eq:appendix-trace-differentiation}.
\end{proof}

The Cauchy--Schwarz inequality gives the boundary trace estimate
\begin{equation}
  |T_\lambda[\Psi]|
  \le
  \|\Psi\|_{L^2(\mu)}
  +\frac{\pi}{\sqrt6}
  \|\partial_x\partial_y\Psi\|_{L^2(\mu)},
  \label{eq:appendix-trace-bound}
\end{equation}
where $\mu$ is product standard Gaussian measure.

For $|t|<1$, one has
\[
  H(t)=T_{-t}[C_{\rho(t)}].
\]
Define
\[
  \Phi_0(t;x,y):=C_{\rho(t)}(x,y),
  \qquad
  \Phi_{k+1}:=D_t\Phi_k-t\partial_x\partial_y\Phi_k.
\]
Applying the differentiation identity three times gives
\begin{equation}
  D^3H(t)=T_{-t}[\Phi_3(t;\cdot,\cdot)].
  \label{eq:appendix-D3-trace}
\end{equation}
The interval bounds below provide a uniform integrable majorant, which
justifies taking radial limits in \eqref{eq:appendix-D3-trace} as
$|t|\uparrow1$.

\subsection{The nine smooth kernels}
\label{app:certification-nine-kernels}

Write $\rho_j=D_t^j\rho$.  For $p+a\ge1$, set
\[
  G_{p,a}(t;x,y)
  :=\left[
  \partial_r^p(\partial_x\partial_y)^aC_r(x,y)
  \right]_{r=\rho(t)}.
\]
Expanding the recursion before taking norms gives
\[
  \Phi_3(t;x,y)
  =\sum_{(p,a)\in\mathcal I}c_{p,a}(t)G_{p,a}(t;x,y),
\]
where
\[
  \mathcal I=
  \{(1,0),(2,0),(3,0),(0,1),(0,2),(0,3),
    (1,1),(2,1),(1,2)\},
\]
and
\begin{center}
\begin{tabular}{c c@{\qquad}c c}
\toprule
$(p,a)$ & $c_{p,a}(t)$ & $(p,a)$ & $c_{p,a}(t)$\\
\midrule
$(1,0)$ & $\rho_3$ & $(0,1)$ & $-t$\\
$(2,0)$ & $3\rho_1\rho_2$ & $(0,2)$ & $3t^2$\\
$(3,0)$ & $\rho_1^3$ & $(0,3)$ & $-t^3$\\
$(1,1)$ & $-3t(\rho_1+\rho_2)$ & $(2,1)$ & $-3t\rho_1^2$\\
$(1,2)$ & $3t^2\rho_1$ & & \\
\bottomrule
\end{tabular}
\end{center}

For $t=e^{i\theta}$ on the unit circle, define the combined energy
\[
  S(\theta)
  :=\bigl\|\Phi_3(e^{i\theta})\bigr\|_{L^2(\mu)}^2
  +4\,\bigl\|\partial_x\partial_y\Phi_3(e^{i\theta})\bigr\|_{L^2(\mu)}^2.
\]
By \eqref{eq:appendix-D3-trace}, the trace estimate
\eqref{eq:appendix-trace-bound}, and the Cauchy--Schwarz inequality
\[
  \Bigl(x+\tfrac{\pi}{\sqrt6}\,y\Bigr)^2
  \le\Bigl(1+\tfrac{\pi^2}{24}\Bigr)\bigl(x^2+4y^2\bigr)
  \le\tfrac{\pi^2}{6}\bigl(x^2+4y^2\bigr),
\]
one has the pointwise bound
$|D^3H(e^{i\theta})|^2\le\tfrac{\pi^2}{6}S(\theta)$.  Since every
construction coefficient is real, $S(-\theta)=S(\theta)$, and therefore
\begin{equation}
  \|D^3H\|_{L^2(\mathbb T)}^2
  =\frac{1}{2\pi}\int_0^{2\pi}
  \bigl|D^3H(e^{i\theta})\bigr|^2\,d\theta
  \le\frac{\pi^2}{6}\cdot\frac{1}{\pi}\int_0^{\pi}S(\theta)\,d\theta.
  \label{eq:appendix-combined-energy}
\end{equation}
The interval computation certifies
\[
  \frac{1}{\pi}\int_0^{\pi}S(\theta)\,d\theta
  \le126.80385221.
\]
Substitution into \eqref{eq:appendix-combined-energy} yields
\begin{equation}
  \|D^3H\|_{L^2(\mathbb T)}
  \le\frac{\pi}{\sqrt6}\,\bigl(126.80385221\bigr)^{1/2}
  <14.443.
  \label{eq:appendix-energy-output}
\end{equation}
Together, \eqref{eq:appendix-head-output} and
\eqref{eq:appendix-energy-output} prove
Proposition~\ref{prop:certified-numerical-input}.

\subsection{Interval-arithmetic implementation}
\label{app:certification-implementation}

For reproducibility, the implementation records the following data.

\begin{enumerate}
  \item Every decimal parameter is treated as an exact terminating rational,
  or is replaced by an explicit outward-rounded interval.
  \item The one-dimensional coefficient integrals are evaluated by adaptive
  interval quadrature on a finite interval, with a closed Gaussian-tail bound
  outside that interval.
  \item Every derivative of $K_r$ is reduced to a polynomially weighted
  Gaussian integral.  The mixed exponential is expanded into an absolutely
  convergent series with a certified tail.
  \item The circle average of $S$ is enclosed by a composite four-point
  Gauss--Legendre rule.  A panelwise eighth-derivative estimate supplies the
  rigorous remainder; by the symmetry $S(-\theta)=S(\theta)$ the average is
  taken over $[0,\pi]$, and the computation uses $1024$ panels.
  \item On every panel, the code verifies the uniform bounds
  $\lvert1-\rho(e^{i\theta})^2\rvert\ge0.3724$ and a Gaussian decay margin of
  at least $0.49$.  This also gives the domination needed for the radial
  boundary passage above.
\end{enumerate}

The certificate terminates by printing the outward-rounded statements
\[
  b_1\ge0.881573822049,
  \qquad
  \sum_{\substack{3\le m\le251\\m\ \mathrm{odd}}}|b_m|
  \le1.1328860\cdot10^{-5},
  \qquad
  \|D^3H\|_{L^2(\mathbb T)}<14.443,
\]
and
\[
  0.881545409+1.1328860\cdot10^{-5}
  +\frac{14.443}{\sqrt{10\cdot251^5}}
  <0.881573822049.
\]
The last inequality is the complete numerical bridge from this appendix to
Section~\ref{sec:certification}.

\section{Certification of the Fiber Inequality}
\label{app:fiber-certificate}

This appendix reduces the fiber inequality~\eqref{eq:fiber-inequality},
\[
        V(u)\le 3\nu\beta(u)-\beta(u)^2
        \qquad\text{for every ternary } u:\R\to\{-1,0,1\},
\]
to a finite family of one-dimensional inequalities, and documents the interval-arithmetic certificates that verify each of them. Throughout, $Z\sim N(0,1)$, $\gamma$ is its law, $\varphi(s)=e^{-s^2/2}/\sqrt{2\pi}$, and
\[
        V(u)=\sum_{j=0}^3\E[u\psi_j]^2,
        \qquad
        \beta(u)=\E\bigl[|Z|\,u(Z)^2\bigr]\in[0,\nu].
\]

All computations are performed in Arb ball arithmetic~\cite{johansson2017arb} (via \texttt{python-flint~0.8.0}, the Python interface to FLINT~3, which incorporates Arb; working precision $256$--$400$ bits): every computed quantity is an interval $[\text{mid}\pm\text{rad}]$ whose radius is a rigorous outward-rounded error bound, and an inequality counts as \emph{certified} only when the entire interval lies on the claimed side of $0$. All certificates are deterministic. The scripts, their logs, and a one-command driver are collected in the repository directory \texttt{lower\_bound\_reproducibility/}; script names below refer to its \texttt{scripts/} folder.

\subsection{Reduction to two certified regimes}
\label{app:fq1-reduction}

The budget axis $\beta\in[0,\nu]$ is covered by two regimes: a direction-free bound for large budgets, and a Fenchel-dual bound for small budgets. Write
\[
        \phi(\beta):=\sqrt{3\nu\beta-\beta^2},
\]
so that~\eqref{eq:fiber-inequality} is exactly $\sqrt{V(u)}\le\phi(\beta(u))$.

\begin{lemma}[High-budget regime]\label{lem:highbudget}
Let $u$ be ternary with $\beta(u)=\beta\in(0,\nu]$, and let $r\in(0,\infty]$ solve $\nu(1-e^{-r^2/2})=\beta$. Then
\[
        V(u)\;\le\;\Pr(u\neq0)\;\le\;2\PhiG(r)-1 .
\]
Consequently, if
\begin{equation}
        g(r):=3\nu\beta(r)-\beta(r)^2-\bigl(2\PhiG(r)-1\bigr)\;\ge\;0,
        \qquad \beta(r):=\nu(1-e^{-r^2/2}),
        \label{eq:g-def}
\end{equation}
then~\eqref{eq:fiber-inequality} holds for $u$.
\end{lemma}

\begin{proof}
The first inequality is Bessel: $\psi_0,\ldots,\psi_3$ are orthonormal, so $V(u)\le\norm{u}_2^2=\E[u^2]=\Pr(u\neq0)$. The second is the bathtub principle already used in Section~\ref{sec:disagreement}: among all supports $E$ with $\int_E|s|\,d\gamma=\beta$, the Gaussian measure $\gamma(E)$ is maximized by the centered interval $E_r=\{|s|\le r\}$ with the same budget. Indeed $(\one_E-\one_{E_r})(|s|-r)\ge0$ pointwise, and integrating gives $0\le 0-r(\gamma(E)-\gamma(E_r))$, i.e.\ $\gamma(E)\le\gamma(E_r)=2\PhiG(r)-1$.
\end{proof}

\begin{lemma}[Dual regime and the Fenchel splice]\label{lem:dual-splice}
Call $p=\sum_{j=0}^3a_j\psi_j$ with $\norm{a}_2=1$ a \emph{unit cubic}, and set
\[
        J(c,p):=\E\bigl[(|p(Z)|-c|Z|)_+\bigr],
        \qquad
        d(c):=\frac{3\nu}{2}\bigl(\sqrt{1+c^2}-c\bigr),
        \qquad
        \beta_*(c):=\frac{3\nu}{2}\Bigl(1-\frac{c}{\sqrt{1+c^2}}\Bigr).
\]
Fix $c_M>0$ and suppose that
\begin{equation}
        J(c,p)\le d(c)
        \qquad\text{for every unit cubic } p \text{ and every } c\ge c_M.
        \label{eq:dual-bound}
\end{equation}
Then~\eqref{eq:fiber-inequality} holds for every ternary $u$ with $\beta(u)\le\beta_*(c_M)$.
\end{lemma}

\begin{proof}
Let $v_j=\E[u\psi_j]$ and $m=\norm{v}_2=\sqrt{V(u)}$. If $m=0$ there is nothing to prove, since $3\nu\beta-\beta^2=\beta(3\nu-\beta)\ge0$ on $[0,\nu]$. Otherwise $p=\sum_j(v_j/m)\psi_j$ is a unit cubic with $m=\E[u\,p]$, and for every $c\ge0$,
\[
        m=\E[u\,p]\le\int_{\{u\neq0\}}|p|\,d\gamma
        \le\int(|p|-c|s|)_+\,d\gamma+c\int_{\{u\neq0\}}|s|\,d\gamma
        =J(c,p)+c\,\beta(u).
\]
Next, $d$ is the concave conjugate of $\phi$: since $\phi(\beta)^2+(\beta-\tfrac{3\nu}{2})^2=(\tfrac{3\nu}{2})^2$, the graph of $\phi$ is the upper semicircle of radius $R=\tfrac{3\nu}{2}$ centered at $(R,0)$, whence
\[
        \sup_{0\le\beta\le2R}\bigl(\phi(\beta)-c\beta\bigr)=R\bigl(\sqrt{1+c^2}-c\bigr)=d(c),
\]
attained at the tangency budget $\beta=\beta_*(c)$. The map $c\mapsto\beta_*(c)$ is a strictly decreasing bijection from $[0,\infty)$ onto $(0,R]$, so for $\beta\le\beta_*(c_M)$ the tangent slope $c(\beta)=\beta_*^{-1}(\beta)$ satisfies $c(\beta)\ge c_M$, and biconjugation on the semicircle gives
\[
        \inf_{c\ge c_M}\bigl(d(c)+c\beta\bigr)=d\bigl(c(\beta)\bigr)+c(\beta)\,\beta=\phi(\beta).
\]
Combining, $m\le\phi(\beta(u))$, i.e.\ $V(u)=m^2\le3\nu\beta(u)-\beta(u)^2$. (For $\beta(u)=0$ the same chain gives $m\le\inf_{c\ge c_M}d(c)=0$, using~\eqref{eq:dual-bound} for arbitrarily large $c$.)
\end{proof}

\begin{remark}[Exactness of the threshold reduction]\label{rem:exactness}
The step $m\le J(c,p)+c\beta(u)$ is the easy half of an exact pointwise duality: for every $z$,
\[
        \max_{v\in\{-1,0,1\}}\bigl(v\,p(z)-c|z|\,v^2\bigr)=\bigl(|p(z)|-c|z|\bigr)_+,
\]
with maximizer $v(z)=\sgn(p(z))\one_{\{|p(z)|\ge c|z|\}}$, so that in fact
\[
        \sup_{u\ \text{ternary}}\bigl(\E[u\,p]-c\,\beta(u)\bigr)=J(c,p),
\]
attained by the threshold rule. The dual gate~\eqref{eq:dual-bound} is therefore not a lossy relaxation of the small-budget regime: it says precisely that every priced threshold optimum lies below the corresponding tangent line of the semicircle.
\end{remark}

\begin{proposition}[Coverage]\label{prop:fq1-coverage}
Take $c_M=0.993405$. The certificates of Sections~\ref{app:cert-highbeta}--\ref{app:cert-highc} establish:
\begin{enumerate}
\item[\textup{(C1)}] $g\ge0$ on $[r_1,\infty]$ with $r_1=1.081596$; hence, by Lemma~\ref{lem:highbudget}, \eqref{eq:fiber-inequality} holds whenever $\beta(u)\ge\beta_0:=\beta(r_1)$, with the certified enclosure $\beta_0=0.353345035094\pm4.4\times10^{-13}$;
\item[\textup{(C2)}, \textup{(C3)}] the dual bound~\eqref{eq:dual-bound} for all $c\in[c_M,12]$ and all $c\ge12$ respectively;
\item[\textup{(S)}] the splice inequality $\beta_0\le\beta_*(c_M)$, with certified enclosures $\beta_*(c_M)=0.353346937306\pm1.1\times10^{-13}$ and $\beta_*(c_M)-\beta_0=1.9022\times10^{-6}\pm3.9\times10^{-14}>0$.
\end{enumerate}
Since $[0,\beta_*(c_M)]\cup[\beta_0,\nu]=[0,\nu]$, the fiber inequality~\eqref{eq:fiber-inequality} holds for every ternary $u$.
\end{proposition}

Everything that remains is finite: (C1) is a one-variable inequality on a half-line, (S) is a single sign check, and (C2)--(C3) are inequalities over compact parameter sets --- the medium regime is certified on $S^3\times[c_M,12]$, and in the tail regime the change of variables $\lambda=1/c$ turns $c\ge12$ into the compact range $S^3\times[0,\tfrac1{12}]$, with the endpoint $\lambda=0$ handled through the scaled Taylor expansion of Section~\ref{app:cert-highc}. The next three subsections describe exactly what is checked in each case.

\subsection{Certificate (C1): the high-budget inequality}
\label{app:cert-highbeta}

Differentiating~\eqref{eq:g-def} and using $2\varphi(r)=\nu e^{-r^2/2}$,
\[
        g'(r)=2\varphi(r)\,B(r),
        \qquad
        B(r)=r\nu\bigl(1+2e^{-r^2/2}\bigr)-1,
        \qquad
        B'(r)=\nu\Bigl(1-2e^{-r^2/2}(r^2-1)\Bigr).
\]
The function $2e^{-r^2/2}(r^2-1)$ is maximized at $r^2=3$, where it equals $4e^{-3/2}$; the certified check $4e^{-3/2}<1$ therefore gives $B'>0$ for all $r\ge1$. Consequently the two thin-ball checks
\[
        g(r_1)>0,\qquad B(r_1)>0 \qquad (r_1=1.081596)
\]
prove $g'\ge0$ hence $g\ge g(r_1)>0$ on all of $[r_1,\infty)$, and the certified limit $g(\infty)=2\nu^2-1>0$ covers $\beta=\nu$. The script \texttt{highbeta\_monotone.py} (prec.\ $300$) certifies all four facts, obtaining
\[
        g(r_1)=3.65\times10^{-4}\pm2.1\times10^{-7},
        \qquad
        B(r_1)=0.824613\pm3.7\times10^{-7},
        \qquad
        4e^{-3/2}=0.89252\ldots<1,
\]
together with an independent (logically redundant) interval scan certifying $B'>0$ on $[1,8)$. A second script, \texttt{highbeta\_refine.py} (prec.\ $400$), supplies the splice constants: it computes the certified enclosures of $\beta_0=\beta(r_1)$ and $\beta_*(c_M)$ quoted in Proposition~\ref{prop:fq1-coverage} and certifies the splice inequality (S), re-verifying $g\ge0$ on $[r_1,\,1.0816]$ with fine interval balls along the way.

\subsection{Certificate (C2): the dual envelope on the medium band $c\in[c_M,12]$}
\label{app:cert-envelope}

Fold $J$ onto the half-line and split $p=e+o$ into its even part $e=a_0\psi_0+a_2\psi_2$ and odd part $o=a_1\psi_1+a_3\psi_3$:
\begin{equation}
        J(c,p)=\int_0^\infty\Bigl[\bigl(|p(s)|-cs\bigr)_++\bigl(|p(-s)|-cs\bigr)_+\Bigr]\varphi(s)\,ds ,
        \label{eq:J-fold}
\end{equation}
and for each $s$ the unordered pair $\{|p(s)|,|p(-s)|\}$ equals $\{|e(s)|+|o(s)|,\;\bigl||e(s)|-|o(s)|\bigr|\}$ (because $|e\pm o|$ take exactly these two values in some order). The integrand of~\eqref{eq:J-fold} is therefore $f\bigl(|e(s)|,|o(s)|\bigr)$ with
\[
        f(x,y):=(x+y-t)_++\bigl(|x-y|-t\bigr)_+,\qquad t=cs\ge0 .
\]

\begin{lemma}[Envelope monotonicity]\label{lem:fold-mono}
For every $t\ge0$, the function $f$ is nondecreasing in each of $x\ge0$ and $y\ge0$. Hence $f\bigl(|e(s)|,|o(s)|\bigr)\le f\bigl(\bar A(s),\omega(s)\bigr)$ for any pointwise upper bounds $\bar A\ge|e|$, $\omega\ge|o|$.
\end{lemma}

\begin{proof}
$f$ is symmetric in $(x,y)$, so it suffices to treat $x$; fix $y$ and distinguish the three regimes of the second term. If $x-y>t$, both terms are active (as $x+y\ge x-y>t$) and the $x$-slope is $1+1=2$. If $y-x>t$, both terms are again active ($x+y>2x+t\ge t$) and the slopes cancel, $1-1=0$. If $|x-y|\le t$, the second term vanishes and $f=(x+y-t)_+$ has $x$-slope $0$ or $1$. In every case the slope is $\ge0$.
\end{proof}

The certificate applies Lemma~\ref{lem:fold-mono} with the Cauchy--Schwarz (Christoffel) envelopes

\begin{align*}
        |o(s)|&\le r_o\sqrt{K_o(s)},
        & K_o&=\psi_1^2+\psi_3^2=\tfrac{5}{2}s^2-s^4+\tfrac16 s^6,\\
        |e(s)|&\le|e_{\bar a}(s)|+\rho\sqrt{K_e(s)},
        & K_e&=\psi_0^2+\psi_2^2=\tfrac32-s^2+\tfrac12 s^4,
\end{align*}
where $\bar a=(\bar a_0,\bar a_2)$ is the center of an even-coefficient box of circumradius $\rho$, and $r_o=\sqrt{1-l^2}$ with $l=\max(0,\norm{\bar a}-\rho)$ bounds the odd-coefficient norm of \emph{any} unit cubic whose even part lies in the box. In this way a single evaluation certifies the box against \emph{all} compatible odd parts, and the search space is the two-dimensional even disk rather than $S^3$ (by $J(c,-p)=J(c,p)$ one may take $a_0\ge0$).

The integral~\eqref{eq:J-fold} is enclosed by $480$ interval panels on $s\in[0,8]$ (per-panel ball enclosures of $f(\bar A,\omega)$ times upper bounds of the panel's Gaussian weight, all roundings outward) plus the closed-form tail majorant
\[
        \int_8^\infty\Bigl(1+\tfrac{s^2+1}{\sqrt2}+s^3\Bigr)\varphi(s)\,ds\;\le\;3.7\times10^{-13}
        \qquad(\text{the integral is }\approx3.64\times10^{-13}).
\]
The majorant is valid because $f(x,y)\le(x+y)+|x-y|=2\max\{x,y\}$, while for $s\ge8$ both $|e|\le1+(s^2+1)/\sqrt2\le s^3/2$ and $|o|\le\sqrt{K_o}\le s^3/\sqrt6\le s^3/2$, so the folded integrand is at most $s^3\le1+\tfrac{s^2+1}{\sqrt2}+s^3$; the scripts use the rigorous Arb enclosure of the tail integral, not the rounded decimal above. On a $c$-band $[c_L,c_U]$ the certificate uses the monotonicities $J(c,p)\le J(c_L,p)$ and $d(c)\ge d(c_U)$, so certifying $\overline{U}(\text{box},c_L)\le\underline{d}(c_U)$ certifies all $c$ in the band at once. The branch-and-bound starts from a grid of spacing $0.08$ on the even disk whose boxes have half-width equal to the spacing (so they cover the disk with twofold overlap), splits any uncertified box into four (down to a failure floor of half-width $6\times10^{-4}$, never reached), and subdivides $c$-bands geometrically on failure.

Run as eight parallel sub-bands whose union is $[c_M,12]$ with no hole (bands~6 and~7 overlap on $[4.0,4.083]$), the certificate resolves $95{,}784$ boxes in about $30$ seconds and certifies~\eqref{eq:dual-bound} on the whole band; per-band results are listed in Table~\ref{tab:fq1-bands}. The worst certified margin over the tiling, $+2.26\times10^{-5}$, measures the tightness of the \emph{envelope relaxation}, not of~\eqref{eq:fiber-inequality} (see the remark at the end of this appendix); it grows under finer $c$-banding.

\begin{table}[h]\centering
\begin{tabular}{lllll}
\hline
band & $c$-range & sub-bands & boxes & worst certified margin \\
\hline
1 & $[0.993405,\,1.19]$ & 12 & 12{,}164 & $+1.088\times10^{-4}$ \\
2 & $[1.19,\,1.30]$ & 14 & 12{,}416 & $+9.058\times10^{-5}$ \\
3 & $[1.30,\,1.45]$ & 12 & 10{,}712 & $+8.958\times10^{-5}$ \\
4 & $[1.45,\,1.75]$ & 8 & 7{,}988 & $+1.956\times10^{-4}$ \\
5 & $[1.75,\,3.5]$ & 24 & 24{,}200 & $+2.260\times10^{-5}$ \\
6 & $[3.5,\,4.083]$ & 10 & 7{,}584 & $+2.830\times10^{-5}$ \\
7 & $[4.0,\,6.0]$ & 16 & 11{,}840 & $+3.571\times10^{-5}$ \\
8 & $[6.0,\,12.0]$ & 14 & 8{,}880 & $+2.539\times10^{-5}$ \\
\hline
\end{tabular}
\caption{Certificate (C2): the eight sub-bands of the dual-envelope branch-and-bound (\texttt{envelope\_band.py}, precision $256$ bits). Each row certifies $J(c,p)\le d(c)$ for \emph{all} unit cubics $p$ and all $c$ in its range; ``sub-bands'' is the geometric $c$-tiling within the row.}
\label{tab:fq1-bands}
\end{table}

\subsection{Certificate (C3): the high-$c$ tail $c\ge12$}
\label{app:cert-highc}

For $c\ge12$ substitute $\lambda=1/c\in(0,\tfrac1{12}]$ and $x=cs$. Both sides of~\eqref{eq:dual-bound} then carry a factor $\lambda$: after dividing it out, the target becomes $D(\lambda)\ge I(\lambda,p)$ with
\[
        D(\lambda)=\frac{3\nu}{2}\cdot\frac{1}{1+\sqrt{1+\lambda^2}},
        \qquad
        I(\lambda,p)=\int_0^\infty\Bigl[\bigl(|p(\lambda x)|-x\bigr)_++\bigl(|p(-\lambda x)|-x\bigr)_+\Bigr]\varphi(\lambda x)\,dx .
\]
As $\lambda\downarrow0$ this degenerates, and it is worth seeing why. For large $c$ the active condition $|p(s)|\ge c|s|$ selects a small neighbourhood of the origin, where $J(c,p)=\varphi(0)\,p(0)^2/c+O(c^{-3})$. Among unit cubics the largest possible value of $p(0)^2$ is the evaluation-kernel norm $K_{\le3}(0,0)=\sum_{j\le3}\psi_j(0)^2=\tfrac32$, attained by the normalized reproducing kernel at the origin,
\[
        p_0(s)=\frac{K_{\le3}(s,0)}{\sqrt{K_{\le3}(0,0)}}=\frac{3-s^2}{\sqrt6}=\sqrt{\tfrac23}\,\psi_0-\tfrac{1}{\sqrt3}\,\psi_2 .
\]
Since $\varphi(0)=\nu/2$, this gives $\sup_pJ(c,p)\sim\tfrac{3\nu}{4c}$, while also $d(c)\sim\tfrac{3\nu}{4c}$: \emph{the dual bound is asymptotically tight}, with equality direction $p_0$. In the scaled variables, $D(0)=I(0,p_0)=\tfrac34\nu$, and the margin vanishes to second order in $\lambda$. This is why the tail cannot be closed by more branch-and-bound---near the limit every box fails---and why the certified object is instead the scaled margin $\bigl(D-I\bigr)/\lambda^2$, which has a strictly positive limit. For orientation (this computation is cross-checked numerically but not used below): in the local chart introduced next,
\[
        \frac{D(\lambda)-I(\lambda,p)}{\lambda^2}
        \;\longrightarrow\;
        \frac{5\nu}{32}+\frac{3\nu}{4}\norm{y}^2
        \qquad(\lambda\downarrow0),
\]
and the certified leading jets of the majorant constructed below reproduce exactly these values minus the far-tail reserve of Lemma~\ref{lem:remote-tail} ($0.094669=\tfrac{5\nu}{32}-0.03$ at $y=0$; increment $0.598414=\tfrac{3\nu}{4}$ at $y=(0,0,1)$), confirming that the majorant is lossless at leading order.

Every unit cubic with $q:=|\ip{p}{p_0}|>0$ can be written (up to sign) in the \emph{local chart}
\[
        p=\frac{p_0+\lambda\,w}{\sqrt{1+\lambda^2\norm{y}^2}},
        \qquad
        w=y_E\,\frac{s^2}{\sqrt3}+y_1\psi_1+y_3\psi_3\perp p_0,
        \qquad
        \norm{w}=\norm{y},
\]
and the transverse coordinate $z:=(1-q^2)/\lambda^2$ satisfies $z=\norm{y}^2/(1+\lambda^2\norm{y}^2)$. The certificate covers all $p$ by two overlapping charts:
\begin{itemize}
\item \textbf{LOCAL} ($\norm{y}\le2$, equivalently $z\le4/(1+4\lambda^2)\ge3.891$ for $\lambda\le\tfrac1{12}$);
\item \textbf{AWAY} ($z\in[3.88,4.2]$ by branch-and-bound in $z$, and $z\ge4.2$ by a certified monotone-domination bracket); this includes $q=0$, where $z=c^2\ge144$.
\end{itemize}
Since $3.88<3.891$ the union covers every $z\ge0$, i.e.\ every unit cubic, for every $c\ge12$.

\subsubsection*{Localization and the far tail}

The active sets $\{x:|p(\pm\lambda x)|\ge x\}$ are \emph{not} compact in the scaled variable: a genuine cubic grows like $\lambda^3x^3$, so the active condition, which for moderate $x$ holds only in a central window, holds again on a remote ray (for $p=\psi_3$ and $\lambda=\tfrac1{12}$ the re-entry is near $x\approx68$). The certificate separates the two components and covers the remote one by an explicit allowance.

\begin{lemma}[Localization]\label{lem:localization}
Let $K(s):=\sum_{j=0}^3\psi_j(s)^2=\tfrac32+\tfrac32s^2-\tfrac12s^4+\tfrac16s^6$. For every unit cubic $p$, every $c\ge12$ and every $s\ge0$: if $|p(s)|\ge cs$, then
\[
        cs<1.5\;(<L:=1.55)
        \qquad\text{or}\qquad
        s>1.55\sqrt c .
\]
\end{lemma}

\begin{proof}
Cauchy--Schwarz gives $|p(s)|^2\le K(s)$, so activity implies $F_c(s^2)\le0$, where $F_c(x):=c^2x-\bigl(\tfrac32+\tfrac32x-\tfrac12x^2+\tfrac16x^3\bigr)$. It therefore suffices that $F_c>0$ on $[(1.5/c)^2,\,1.55^2c]$. Since $F_c'(x)=c^2-\tfrac32+x-\tfrac{x^2}2$ is a downward parabola in $x$, positive at the left endpoint, $F_c$ is increasing-then-decreasing on the interval, so its minimum is attained at an endpoint. At the left endpoint, $F_c\bigl((1.5/c)^2\bigr)=\tfrac34-\tfrac32u+\tfrac{u^2}2-\tfrac{u^3}6$ with $u=2.25/c^2\in(0,\tfrac1{64}]$, certified positive over the whole $u$-range; at the right endpoint, $G(c):=F_c(1.55^2c)$ is a cubic in $c$ with positive leading coefficient $1.55^2(1-1.55^4/6)$ and certified $G(12)>0$, $G'(12)>0$ (and $G'$ is increasing), so $G>0$ for all $c\ge12$. The three sign evaluations are certified in \texttt{highc\_remote\_tail.py}.
\end{proof}

\begin{lemma}[Far-tail allowance]\label{lem:remote-tail}
For every unit cubic $p$ and every $c\ge12$,
\[
        J_{\mathrm{rem}}(c,p):=\int_{\{|s|\ge1.55\sqrt c\}}\bigl(|p(s)|-c|s|\bigr)_+\,d\gamma
        \;\le\;2\bigl(1.55^2c+2\bigr)\,\varphi\bigl(1.55\sqrt c\bigr)
        \;\le\;\frac{0.03}{c^3};
\]
equivalently, the remote active components contribute at most $0.03\lambda^2$ to $I(\lambda,p)$.
\end{lemma}

\begin{proof}
For $s\ge2$ one has $K(s)\le s^6$ (equivalent to $9+9s^2\le3s^4+5s^6$, true since $3s^4\ge12s^2>9s^2+9$), so on the remote set $\bigl(|p(s)|-c|s|\bigr)_+\le\sqrt{K(s)}\le s^3$; applying $\int_a^\infty s^3\varphi\,ds=(a^2+2)\varphi(a)$ on both half-lines gives the first bound. For the second, set $t(c):=c^3\cdot2(1.55^2c+2)\varphi(1.55\sqrt c)$; then
\[
        (\ln t)'(c)=\frac3c+\frac{1.55^2}{1.55^2c+2}-\frac{1.55^2}2
        \;\le\;\frac3{12}+\frac{1.55^2}{12\cdot1.55^2+2}-\frac{1.55^2}2\;<\;-0.87,
\]
so $t$ is decreasing on $[12,\infty)$; the certified value $t(12)=0.0233\ldots<0.03$ completes the proof. Both sign checks are certified in \texttt{highc\_remote\_tail.py}.
\end{proof}

This is the exact role of the reserve $0.03\lambda^2$ deducted inside the certified functionals below: it pays for the far tail, so that certifying the central majorant against $D-0.03\lambda^2$ certifies all of~\eqref{eq:dual-bound}. The central window is then frozen at the fixed cutoff $A:=p_0(0)=\sqrt{3/2}$ (the limiting active boundary, since $p_0(0)=x$ at $x=A$) by the following device, which avoids an interval-unstable moving root:

\begin{lemma}[Moving cutoff]\label{lem:moving-cutoff}
Let $0\le w(x)\le\varphi_0:=1/\sqrt{2\pi}$, let $\delta:[0,L]\to\R$ be measurable with $|\delta|\le\Delta$, and let $0<A\le L$. Then
\[
        \int_0^L\bigl(A-x+\delta(x)\bigr)_+\,w(x)\,dx
        \;\le\;
        \int_0^A\bigl(A-x+\delta(x)\bigr)\,w(x)\,dx+\varphi_0\Delta^2 .
\]
\end{lemma}

\begin{proof}
On $[0,A]$, $(A-x+\delta)_+-(A-x+\delta)=(x-A-\delta)_+\le(x-A+\Delta)_+$, which is supported on $[A-\Delta,A]$ and bounded by $\Delta$; its $w$-integral is at most $\varphi_0\Delta^2/2$. On $[A,L]$, $(A-x+\delta)_+\le(\Delta-(x-A))_+$, again contributing at most $\varphi_0\Delta^2/2$.
\end{proof}

\subsubsection*{The LOCAL majorant}

In the local chart, with $n:=\sqrt{1+\lambda^2\norm{y}^2}$,
\[
        n\,p(\pm\lambda x)=P_\pm(x)
        :=A-\frac{\lambda^2x^2}{\sqrt6}+\frac{y_E\lambda^3x^2}{\sqrt3}
        \pm\lambda^2\Bigl(y_1x+y_3\,\frac{\lambda^2x^3-3x}{\sqrt6}\Bigr).
\]
An interval evaluation certifies $P_\pm/n>0$ for $x\in[0,L]$, $\lambda\in[0,\tfrac1{12}]$ and the whole $y$-range, so there $|p(\pm\lambda x)|=P_\pm(x)/n$; by Lemmas~\ref{lem:localization} and~\ref{lem:remote-tail},
\begin{equation}
        I(\lambda,p)\;\le\;\int_0^L\Bigl[\Bigl(\frac{P_+}{n}-x\Bigr)_++\Bigl(\frac{P_-}{n}-x\Bigr)_+\Bigr]\varphi(\lambda x)\,dx
        \;+\;0.03\lambda^2 .
        \label{eq:reserve}
\end{equation}
Write $P_\pm/n-x=A-x+\delta_\pm(x)$ with $\delta_\pm:=P_\pm/n-A$, and apply Lemma~\ref{lem:moving-cutoff} to each branch with weight $w=\varphi(\lambda\,\cdot)\le\varphi_0$. In $\sum_\pm\int_0^A\delta_\pm w$ the $y_1,y_3$ terms cancel (they are odd under the branch sign), leaving, with the frozen moments $W_j(\lambda):=\int_0^Ax^j\varphi(\lambda x)\,dx$,
\[
        \int_0^L\Bigl[\,\cdots\Bigr]\varphi(\lambda x)\,dx
        \;\le\;U_{\mathrm{loc}}
        :=2\bigl(AW_0-W_1\bigr)+2A\Bigl(\frac1n-1\Bigr)W_0
        -\frac{2\lambda^2}{\sqrt6\,n}W_2+\frac{2y_E\lambda^3}{\sqrt3\,n}W_2
        +\varphi_0\bigl(\Delta_+^2+\Delta_-^2\bigr),
\]
where $\Delta_\pm\ge\sup_{[0,L]}|\delta_\pm|$ is supplied by the split $\Delta_\pm=M_E+M_O$ into branch-even and branch-odd parts,
\[
        M_E=A\Bigl(1-\frac1n\Bigr)+\frac{\lambda^2L^2}{\sqrt6\,n}+\frac{|y_E|\,\lambda^3L^2}{\sqrt3\,n},
        \qquad
        M_O=\frac{\lambda^2\norm{y}L}{n}\sqrt{\tfrac52}\,.
\]
$M_E$ bounds the branch-even part of $\delta_\pm$ termwise on $[0,L]$; $M_O$ bounds the branch-odd part by Cauchy--Schwarz,
\[
        \Bigl|y_1x+y_3\,\frac{\lambda^2x^3-3x}{\sqrt6}\Bigr|
        =\frac{\bigl|y_1\psi_1(\lambda x)+y_3\psi_3(\lambda x)\bigr|}{\lambda}
        \le\frac{\norm{y}}{\lambda}\sqrt{K_o(\lambda x)}
        \le x\norm{y}\sqrt{\tfrac52},
\]
using $K_o(s)=s^2\bigl(1+\tfrac{(3-s^2)^2}{6}\bigr)\le\tfrac52s^2$ for $s^2\le6$; an interval check certifies $s^2=\lambda^2x^2\le0.017$ on the chart, so the condition holds with a wide margin.

\subsubsection*{The AWAY majorant}

Here only the size of the transverse part enters, not its direction. Write $p=\pm(q\,p_0+\tau)$ with $\tau\perp p_0$, $q=\sqrt{1-\lambda^2z}$ and $\norm{\tau}=\lambda\sqrt z$, and let
\[
        K_\perp(s):=K(s)-p_0(s)^2=\tfrac52s^2-\tfrac23s^4+\tfrac16s^6\;\le\;\tfrac52s^2
        \qquad(s^2\le4).
\]
Then $|\tau(\pm\lambda x)|\le\lambda\sqrt z\,\sqrt{K_\perp(\lambda x)}\le\lambda^2x\sqrt{5z/2}$, so both branches obey the same bound $|p(\pm\lambda x)|-x\le A-x+\delta_z(x)$ with
\[
        \delta_z(x)=A(q-1)-\frac{q\lambda^2x^2}{\sqrt6}+\lambda^2x\sqrt{\tfrac{5z}2}
\]
(no branch-positivity check is needed here, because $|p|$ is bounded from above directly). Lemma~\ref{lem:moving-cutoff} applied to both branches then gives
\[
        U_{\mathrm{away}}
        =2\bigl(AW_0-W_1\bigr)+2\Bigl(A(q-1)W_0-\frac{q\lambda^2}{\sqrt6}W_2+\lambda^2\sqrt{\tfrac{5z}2}\,W_1\Bigr)
        +2\varphi_0M_z^2,
\]
with $M_z=A(1-q)+\dfrac{\lambda^2L^2}{\sqrt6}+\lambda^2L\sqrt{\dfrac{5z}2}\;\ge\;\sup_{[0,L]}|\delta_z|$ (using $q\le1$). Finally, $z\ge4.2$ is reduced to the certified band by domination:

\begin{lemma}[Domination for $z\ge z_0=4.2$]\label{lem:domination}
Let $a_z(x):=\sqrt{1-\lambda^2z}\,p_0(\lambda x)+\lambda\sqrt z\,\sqrt{K_\perp(\lambda x)}$ denote the amplitude bound above, and let
\[
        \widetilde K_\perp(\lambda,x):=\frac{K_\perp(\lambda x)}{\lambda^2}
        =\tfrac52x^2-\tfrac23\lambda^2x^4+\tfrac16\lambda^4x^6,
\]
a polynomial in $(\lambda,x)$ (the continuous extension to $\lambda=0$). If
\[
        z_0\,p_0(\lambda x)^2\;\ge\;\widetilde K_\perp(\lambda,x)\,\bigl(1-z_0\lambda^2\bigr)
        \qquad\text{for all }\lambda\in[0,\tfrac1{12}],\ x\in[0,L],
\]
then $\partial_za_z(x)\le0$ for every $z\ge z_0$ and $\lambda\in(0,\tfrac1{12}]$; hence every $z\ge z_0$ (including $q=0$, i.e.\ $z=c^2$) is pointwise dominated by $z=z_0$, which lies inside the certified band $[3.88,4.2]$.
\end{lemma}

\begin{proof}
For $\lambda>0$, $\partial_za_z=-\dfrac{\lambda^2p_0}{2\sqrt{1-\lambda^2z}}+\dfrac{\lambda\sqrt{K_\perp}}{2\sqrt z}\le0$ iff $\lambda^2z\,p_0^2\ge(1-\lambda^2z)K_\perp(\lambda x)$, i.e.\ (dividing by $\lambda^2$) iff $z\,p_0^2\ge\widetilde K_\perp\,(1-z\lambda^2)$. The displayed hypothesis is this condition at $z=z_0$; as $z$ grows the left side increases and the right side decreases, so it persists for all $z\ge z_0$. The hypothesis is certified on an interval grid with worst lower bound $+0.29$ (\texttt{highc\_certificate.py}).
\end{proof}

\subsubsection*{Evaluation and results}

Each frozen moment is an entire function of $\lambda^2$,
\[
        W_j(\lambda)=\varphi_0\sum_{m\ge0}\frac{(-1)^m}{2^m\,m!}\,\frac{A^{\,j+2m+1}}{j+2m+1}\,\lambda^{2m},
\]
evaluated through the terms $m\le9$, and the omitted tail must be accounted for \emph{with its power of $\lambda$}. It is an alternating series with decreasing terms (successive ratios are at most $\lambda^2A^2/2\le\tfrac1{192}$), hence dominated by its first term:
\[
        \bigl|W_j(\lambda)-\widehat W_j(\lambda)\bigr|
        \;\le\;\varphi_0\,\frac{(A^2/2)^{10}}{10!}\,\frac{A^{\,j+1}}{j+21}\,\lambda^{20}
        \;\le\;10^{-9}\,\lambda^{20}
        \qquad(0\le j\le3),
\]
where $\widehat W_j$ denotes the truncation. The $W_j$ enter $U_{\mathrm{loc}}$ and $U_{\mathrm{away}}$ affinely, with coefficient magnitudes summing to less than $10$, so the truncation perturbs $G$ by at most $10^{-8}\lambda^{20}$ and the scaled margin $R=G/\lambda^2$ by at most $10^{-8}\lambda^{18}\le10^{-27}$, \emph{uniformly} on $(0,\tfrac1{12}]$: the $\lambda^{20}$ factor survives the division by $\lambda^2$, so the comparison is scale-aware and negligible against the certified margins below (all $\ge4\times10^{-3}$). With these ingredients, $G:=D-U_{\mathrm{loc/away}}-0.03\lambda^2$ is an explicit algebraic expression in $\lambda$, and a rigorous lower bound for $R:=G/\lambda^2$, \emph{uniform in} $\lambda\in(0,\tfrac1{12}]$, is obtained from the Taylor jet of $G$ at $\lambda=0$ (Arb series coefficients) plus a Lagrange remainder enclosed by the corresponding ball-centered jet over $[0,\tfrac1{12}]$, for each interval box of chart parameters. By~\eqref{eq:reserve} and its AWAY analogue, certifying $R>0$ certifies $D-I>0$, i.e.~\eqref{eq:dual-bound}, on the box.

The branch-and-bound (\texttt{highc\_certificate.py}, precision $300$ bits) certifies: LOCAL, $R>0$ over the ball $\norm{y}\le2$ ($1{,}000$ certified boxes, $1{,}319$ subdivisions, $\approx1$ s), worst certified value $+4.211\times10^{-3}$; AWAY, $R>0$ on $z\in[3.88,4.2]$ (certified in a single interval evaluation), worst $+3.537\times10^{-2}$, together with the domination bracket of Lemma~\ref{lem:domination} and the range check for $M_O$; and the localization and far-tail constants of Lemmas~\ref{lem:localization}--\ref{lem:remote-tail} are certified in \texttt{highc\_remote\_tail.py}.

\subsection{Summary of certified results and reproducibility}
\label{app:cert-summary}

\begin{table}[h]\centering
\footnotesize
\setlength{\tabcolsep}{4pt}
\begin{tabular}{llll}
\hline
certificate & script(s) & what is certified & worst margin \\
\hline
(C1) point $+$ monot. & \texttt{highbeta\_monotone.py} & $g(r_1)>0$, $B(r_1)>0$, $4e^{-3/2}<1$, $g(\infty)>0$ & $+3.65\times10^{-4}$ \\
(C1)$+$(S) fine scan & \texttt{highbeta\_refine.py} & $g\ge0$ on $[r_1,1.0816]$; $\beta_0<\beta_*(c_M)$ & $+1.90\times10^{-6}$ \\
(C2) medium band & \texttt{envelope\_band.py} $\times8$ & $J(c,p)\le d(c)$, all $p\in S^3$, $c\in[c_M,12]$ & $+2.26\times10^{-5}$ \\
(C3) tail, LOCAL & \texttt{highc\_certificate.py} & $R>0$, $\norm{y}\le2$, $\lambda\in(0,\tfrac1{12}]$ & $+4.21\times10^{-3}$ \\
(C3) tail, AWAY & \texttt{highc\_certificate.py} & $R>0$, $z\in[3.88,4.2]$; domination $z\ge4.2$ & $+3.54\times10^{-2}$ \\
(C3) far tail & \texttt{highc\_remote\_tail.py} & active: $cs<1.5$ or $s>1.55\sqrt c$; $J_{\mathrm{rem}}\le0.03c^{-3}$ & $+6.66\times10^{-3}$ \\
\hline
\end{tabular}
\caption{The fiber-inequality certificate, region by region (run of 2026-08-04; all margins are rigorous one-sided ball bounds). Together with Lemmas~\ref{lem:highbudget} and~\ref{lem:dual-splice} these establish~\eqref{eq:fiber-inequality} for every ternary $u$ via Proposition~\ref{prop:fq1-coverage}.}
\label{tab:fq1-results}
\end{table}

Table~\ref{tab:fq1-results} summarizes the certificate. The complete suite runs in under two minutes on one machine via the driver \texttt{run\_all.sh} of \texttt{lower\_bound\_reproducibility/}, which executes every script, stores the logs under \texttt{certificates/}, and verifies the certified verdict line of each log; the only requirement is Python~$\ge3.11$ with \texttt{python-flint~0.8.0}. Every certificate is deterministic: re-runs in the pinned environment reproduce the printed margins verbatim, and independent environments must reproduce the same certified verdicts and enclosures.

\begin{remark}[The certified margins versus the true margin]
The small margins in Table~\ref{tab:fq1-results} measure the slack of the \emph{relaxations} being certified (the support bound of Lemma~\ref{lem:highbudget} and the dual envelope of Lemma~\ref{lem:dual-splice} are both nearly tight along their splice), not the distance of~\eqref{eq:fiber-inequality} from failure. In the extremal coordinates of Section~\ref{sec:disagreement}, the constant fiber $u\equiv1$ (for which $V=1$ and $\beta=\nu$) \emph{realizes}
\[
        \rho\;=\;\frac{3\nu^2-1}{\nu^2}\;=\;3-\frac{\pi}{2}\;=\;1.42920\ldots,
\]
and numerical scans over the threshold family $(a,\tau)$ found no smaller value, suggesting that $3-\pi/2$ is the true infimum of $\rho$; we neither need nor prove that sharper assertion. The fiber inequality asserts only $\rho\ge1$, so it holds with about $0.43$ of numerical headroom; in particular, the sharper inequality $V\le3\nu\beta-c\beta^2$ appears to hold for every $c<3-\pi/2$, though only $c=1$ is certified and used here.
\end{remark}